\documentclass[12pt,english]{article}
\usepackage{mathptmx}
\usepackage[T1]{fontenc}
\usepackage[latin9]{inputenc}
\usepackage{geometry}
\usepackage{xcolor}
\usepackage{babel}
\usepackage{units}
\usepackage{amsmath}
\usepackage{amsthm}
\usepackage{amssymb}
\usepackage{graphicx}
\usepackage{setspace}
\usepackage[authoryear]{natbib}
\usepackage[unicode=true,pdfusetitle,
 bookmarks=true,bookmarksnumbered=false,bookmarksopen=false,
 breaklinks=false,pdfborder={0 0 0},pdfborderstyle={},backref=false,colorlinks=true]
 {hyperref}
\hypersetup{
 citecolor=blue, linkcolor=blue, urlcolor=blue}

\makeatletter
\theoremstyle{remark}
\newtheorem{rem}{\protect\remarkname}
\theoremstyle{definition}
\newtheorem{defn}{\protect\definitionname}
\theoremstyle{plain}
\newtheorem{prop}{\protect\propositionname}
\theoremstyle{plain}
\newtheorem{cor}{\protect\corollaryname}
\theoremstyle{plain}
\newtheorem{thm}{\protect\theoremname}
\theoremstyle{definition}
 \newtheorem{example}{\protect\examplename}
\theoremstyle{plain}
\newtheorem{lem}{\protect\lemmaname}

\usepackage{lmodern}
\usepackage[T1]{fontenc}
\definecolor{brown}{RGB}{0, 0, 0}

\makeatother

\providecommand{\corollaryname}{Corollary}
\providecommand{\definitionname}{Definition}
\providecommand{\examplename}{Example}
\providecommand{\lemmaname}{Lemma}
\providecommand{\propositionname}{Proposition}
\providecommand{\remarkname}{Remark}
\providecommand{\theoremname}{Theorem}

\begin{document}
\title{Social Welfare in Search Games with Asymmetric Information}
\author{Gilad Bavly, Yuval Heller, Amnon Schreiber\thanks{Department of Economics, Bar-Ilan University, Israel. Email addresses:
\protect\href{mailto:gilad.bavly@gmail.com}{gilad.bavly@gmail.com},
\protect\href{mailto:yuval.heller@biu.ac.il}{yuval.heller@biu.ac.il},
\protect\href{mailto:amnon.schreiber@biu.ac.il}{amnon.schreiber@biu.ac.il}.
This manuscript replaces an obsolete working paper titled ``The social
payoff in differentiation games.'' We thank Sergiu Hart, Igor Letina,
Igal Milchtaich, Abraham Neyman, Ron Peretz, Alon Raviv, Dov Samet,
and Eilon Solan, the editor Tilman B$\ddot{\textrm{o}}$rgers and
two anonymous referees for various helpful comments. Bavly and Heller
are grateful to the European Research Council for its financial support
(ERC starting grant \#677057). Bavly acknowledges support from the
Department of Economics and the Sir Isaac Wolfson Chair at Bar-Ilan
University, and ISF grant 1626/18 and 2566/20. Schreiber acknowledges
support from ISF grant 2897/20. }}
\maketitle
\begin{abstract}
\noindent We consider games in which players search for a hidden prize,
and they have asymmetric information about the prize's location. We
study the social payoff in equilibria of these games. We present sufficient
conditions for the existence of an equilibrium that yields the first-best
payoff (i.e., the highest social payoff under any strategy profile),
and we characterize the first-best payoff. The results have interesting
implications for innovation contests and R\&D races.
Final preprint of a manuscript accepted for publication in \emph{Journal
of Economic Theory}.
\end{abstract}

\section{Introduction }

We study interactions in which each agent has private information
about the values of choosing different alternatives, and the payoff
that each agent gains from each alternative is decreasing in the number
of opponents who also choose it. The environment is competitive in
the sense that the agents work separately, have different goals, and
do not share their information. Yet, the agents also share a common
coordination incentive: they may all gain from dividing the different
alternatives between them, thereby avoiding, as much as possible,
cases of miscoordination, i.e., cases of multiple agents choosing
the same alternative, and reducing their own payoffs due to this inefficiency.

An important class of these types of interactions is search games,
in which players search for a prize hidden in one of a finite set
of locations.\footnote{The assumption of having a single prize is common in the literature;
see, e.g., \citet{fershtman1997simple,konrad2014search,liu2019strategic}.
We leave the important question of how to extend the results to multiple
prizes for future research.} Each player $i$ is able to search in at most $K_{i}$ locations
(all at once). Searching incurs a private cost, which is a convex
\textcolor{brown}{increasing }function of the number of locations
in which the player searches. Each player receives some private coarse
signal about the actual location of the prize, and chooses which locations
to search. Specifically, for each player there is a collection of
disjoint subsets of locations (namely, a partition), such that her
private signal informs the player in which of these subsets the prize
resides. 

Both the discoverer and society gain from the discovery of the prize.
We allow the prize's value to depend on the location. Also, the value
for society and the individual values for players may all be different.
When multiple players search in the same location (``search duplication''),
it reduces the reward that each player will receive in case the prize
is indeed there. By contrast, the social value of the prize is unaffected
by the number of finders.

Search games have been applied in various setups such as R\&D races
in oligopolistic markets (e.g., \citealp{loury1979market,chatterjee2004rivals,akcigit2015role,letina2016road}),
design of innovation contests (e.g., \citealp{erat2012managing,bryan2017direction,letina2019inducing}),
and scientific research (e.g., \citealp{kleinberg2011mechanisms}).
Most of the existing literature does not allow for private information.\footnote{\label{fn:Chen-et-al}We are aware of one related existing model of
a search game with asymmetric information, that of \citet{chen2015fair}.
The key difference between our model and theirs is that \citeauthor{chen2015fair}
rely on enforceable mechanisms, which allow players to safely share
their asymmetric information, as all players must follow a contract
once it has been signed. By contrast, we consider a setup in which
players cannot rely on enforceable mechanisms, and, thus, they are
limited to playing Nash equilibria.} The main methodological innovation of the present model is the introduction
of asymmetric information into search games.

Importantly, we study a one-shot game (i.e., if the prize is not found,
players do not get to search again) with simultaneous actions. This
assumption, which differs from the dynamic models studied in many
of the papers cited above, may be reasonable in situations in which
there is severe urgency to make the discovery (see Section \ref{sec:Discussion}
for further discussion, and Section \ref{subsec:Necessity-of-All}
for \textcolor{brown}{an example} of what happens when this assumption
is relaxed). \textcolor{brown}{One recent real-life example in which
urgency might make the interaction essentially one-shot (and which,
roughly, fits the other assumptions of our model) is the problem faced
in 2020 of quickly developing a vaccine for COVID-19, where different
pharmaceutical R\&D divisions had heterogeneous private information
about the most promising route to achieve this. }

The expected social gain (from a successful discovery) in search games
is clearly constrained by the information structure, as we assume
that players are competitive and do not share their private information.
The social gain may also be constrained by the fact that players'
individual preferences can differ from society's, and players have
strategic considerations as well. Thus, the main question we study
is: what is the highest social payoff in equilibrium? 

Our first main result  states that there exists a (pure) equilibrium
that yields the first-best social payoff (namely, the highest social
payoff that any strategy profile can yield) if the following two conditions
hold for any two locations $\omega$ and $\omega'$ that a player
considers possible (after observing her own private signal): (1) ordinal
consistency: the player and society have the same ordinal ranking
between searching (by herself) in $\omega$ and in $\omega'$, and
(2) solitary-search dominance: the player always prefers searching
$\omega$ by herself to searching $\omega'$ with other players, or
to not searching at all. 

It is relatively easy to see that neither condition can be dropped
(see the examples presented in Section \ref{subsec:Necessity-of-All}),
and that the conditions are sufficient in a simple setup without asymmetric
information. Our result shows that, perhaps surprisingly, these two
conditions are sufficient in the richer setup with asymmetric information
as well. The intuition is that no player has an incentive to ``spoil''
society's payoff by moving from a socially better location to a worse
one, nor by moving from a location that she searches alone to a location
that others search. We discuss the implications of this result on
the design of innovation contests in Section \ref{subsec:Insights-for-Innovation}.

Our second main result presents lower bounds for the first-best social
payoff, which we demonstrate to be binding in various cases. The derivation
of these bounds relies on representing a search game as a bipartite
graph and adapting and extending classic results from graph theory,
the max-flow min-cut theorem (\citealp{ford1956maximal}) and the
Birkhoff--von Neumann theorem (\citealp{birkhoff1946tres,von1953certain}),
to our setup.\footnote{Recent economic applications (and extensions) of these graph-theory
results have appeared in matching mechanisms (e.g., \citealp{budish2013designing,bronfman2018redesigning}),
large anonymous games (e.g., \citealp{blonski2005women}), public
good games with multiple resources (e.g., \citealp{tierney2019problem}),
and auctions of multiple discrete items (e.g., \citealp{ben2017walrasian}).} 

\paragraph{Structure}

Section \ref{sec:Model} presents our model. We study the existence
of an equilibrium with a first-best social payoff in Section \ref{sec:Socially-Optimal-Equilibrium}.
In Section \ref{sec:socially-optimal-payoff} we present lower bounds
for the first-best payoff. Section \ref{subsec:Cost-Inclusive-Social-Payoff}
considers a variant of our model, in which society internalizes the
players' costs. We conclude and discuss the relations with the literature
in Section \ref{sec:Discussion}. \textcolor{brown}{The appendix}
presents the formal proofs. 

\section{Model\label{sec:Model}}

\paragraph{Setup}

Let $N=\left\{ 1,2,...,n\right\} $ be a finite set of players. A
typical player is denoted by $i$. We use $-i$ to denote the set
of all players except player $i$. We describe the private information
of the players in terms of knowledge partitions (\citealp{aumann1976agreeing}).
Let $\Omega$ be the set of the states of the world (henceforth, states).
Nature chooses one state $\omega\in\Omega$ that is the true state
of the world. Each player $i$ is endowed with $\Pi_{i}$, which is
a partition of $\Omega$, namely, a list of disjoint subsets of $\Omega$
whose union is the whole $\Omega$. We refer to the elements of player
$i$'s partition (i.e., the subsets) as player $i$'s cells. For each
state $\omega$, let $\pi_{i}\left(\omega\right)$ denote the cell
of player $i$ that contains the state $\omega$. If the true state
is $\omega,$ then player $i$ knows that the true state is one of
the states in $\pi_{i}\left(\omega\right)$. 

\global\long\def\sOne{s^{{\scriptscriptstyle 1}}}%
\global\long\def\sKay{s^{{\scriptscriptstyle K}}}%
\global\long\def\sSubOne{s_{{\scriptscriptstyle 1}}}%
\global\long\def\sTee{s^{{\scriptscriptstyle T}}}%
Note that the knowledge partitions framework is equivalent to a model
in which each player observes a private random signal. Each cell of
player $i$'s partition corresponds to a different realization of
her private signal. W.l.o.g., one may view the partition $\Pi_{i}$
as the set of possible realizations of player $i$'s private signal
itself; i.e., each cell in $\Pi_{i}$ is a possible signal, and if
the state of the world is $\omega$ then player $i$ observes the
signal $\pi_{i}\left(\omega\right)$.

The players search for a prize hidden in one of a finite set of possible
locations. Importantly, we assume that the location of the prize determines
the private signal of each player (in other words, the signals that
players observe are a deterministic function of the prize's location).
This implies that w.l.o.g. each state of the world in our model corresponds
to a different location of the prize. Hence, we identify the finite
set of locations with the set of states $\Omega$. When a player searches
in location (i.e., state) $\omega\in\Omega$, she finds the prize
if the location of the prize is $\omega$ (i.e., if the true state
of the world is $\omega$). Figure \ref{fig:Illustration-of-Information}
demonstrates an information structure in a two-player search game.
\begin{figure}[h]
\centering{}\includegraphics[scale=0.43]{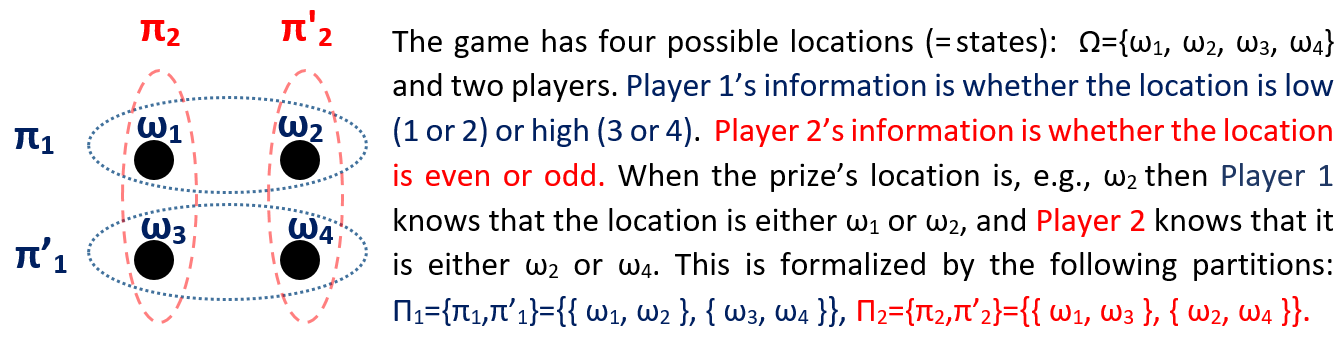}\caption{\label{fig:Illustration-of-Information}Illustration of information
structure of a two-player search game}
\end{figure}

We say that player $i$ receives no information at all if her information
partition $\Pi_{i}$ is trivial, i.e., $\Pi_{i}=\left\{ \Omega\right\} $
contains a single element, which is the whole $\Omega$. A setting
that does not allow for asymmetric information corresponds to the
degenerate case in our model where all players have trivial partitions.

Let $\mu\in\Delta\left(\Omega\right)$ denote the (common) prior belief
about the prize's location, where $\Delta\left(\Omega\right)$ denotes
the set of distributions over $\Omega$. For a subset of locations
$E\subseteq\Omega$, let $\mu\left(E\right)=\sum_{\omega\in E}\mu\left(\omega\right)$
denote the prior probability of $E$. For non-triviality, we assume
that every cell has a positive prior probability, i.e., $\mu\left(\pi_{i}\right)>0$
for every $\pi_{i}\in\Pi_{i}$ and $i\in N$. When the (unknown) location
of the prize is $\omega$, each player $i$ assigns a posterior belief
of $\mu\left(\omega'|\pi_{i}\left(\omega\right)\right)$ to the location
being $\omega'$, where 
\[
\mu\left(\omega'|\pi_{i}\left(\omega\right)\right)=\begin{cases}
\nicefrac{\mu\left(\omega'\right)}{\mu\left(\pi_{i}\left(\omega\right)\right)} & \omega'\in\pi_{i}\left(\omega\right)\\
0 & \omega'\not\in\pi_{i}\left(\omega\right).
\end{cases}
\]

We allow heterogeneity in the maximal number of locations that each
player can search. Specifically, each player $i$ chooses up to $K_{i}\in\mathbb{N}$
locations in which she searches, where $K_{i}$ is the player's search
capacity. A (pure) strategy of player $i$ is a function $s_{i}$
that assigns to each cell $\pi_{i}\in\Pi_{i}$ a subset of $\pi_{i}$
with at most $K_{i}$ elements. We interpret $s_{i}\left(\pi_{i}\right)$
as the set of up to $K_{i}$ locations in which player $i$ searches
when she observes the signal $\pi_{i}$. If no ambiguity can arise,
we may also say that player $i$ (ex-ante) searches in location $\omega$,
if $\omega\in s_{i}\left(\pi_{i}\left(\omega\right)\right)$, i.e.,
if player $i$ searches in $\omega$ when the prize is located in\footnote{Equivalently, the locations in which player $i$ (ex-ante) searches
are $\cup_{\pi_{i}\in\Pi_{i}}\,s_{i}\left(\pi_{i}\right)$, namely,
the union (across all her cells) of the locations she searches within
each cell when that cell happens to be her signal.} $\omega$.

We focus in the present paper on pure strategies. Let $S_{i}\equiv S_{i}\left(G\right)$
denote the set of all (pure) strategies of player $i$, and let $S\equiv S\left(G\right)=\prod_{i\in N}S_{i}$
\,be the set of strategy profiles in the game $G$. For example,
in Figure \ref{fig:Illustration-of-Information} Player $1$, with
a capacity of one, has four pure strategies. One such strategy, denoted
by $\sSubOne$, is given by $\left(\sSubOne\left(\pi_{1}\right)=\omega_{2};\,\sSubOne\left(\pi'_{1}\right)=\omega_{3}\right)$;
i.e., a player following $\sSubOne$ searches in location $\omega_{2}$
upon observing signal $\pi_{1}$ and searches in $\omega_{3}$ upon
observing $\pi'_{1}$. Suppose that Player $1$ follows $\sSubOne$
and the location of the prize is $\omega_{4}$. Then she will observe
the signal $\pi'_{1}$ and search in $\omega_{3}$ (and hence she
will not find the prize).
\begin{rem}
All of our results hold in a more general setup with either of the
following extensions (with minor modifications to the proofs):
\begin{enumerate}
\item Heterogeneous priors: each player $i$ has a different prior $\mu_{i}$. 
\item Heterogeneous restricted locations: each player $i$ is allowed to
search only in a subset $\Omega_{i}\subseteq\Omega$ of the locations.
\end{enumerate}
\end{rem}

\paragraph{Costs, Rewards, and Duplication}

Searching incurs a private cost, which is a convex increasing function
of the number of locations in which a player searches.\footnote{Extending the costs to depend also on which locations, not just how
many, are being searched may be an interesting direction for future
research.} Specifically, each player $i$ bears a cost $c_{i}\left(k\right)\geq0$
when searching within $k$ locations, where $c_{i}\left(0\right)=0$
and $c_{i}\left(k+1\right)-c_{i}\left(k\right)\geq c_{i}\left(k\right)-c_{i}\left(k-1\right)$
for any $k\in\left\{ 1,..,K_{i}-1\right\} $. We say that a game has
a costless search (up to the capacity constraints) if $c_{i}\equiv0$
(i.e., if $c_{i}\left(k\right)=0$ for every $k\in\left\{ 1,..,K_{i}\right\} $
and every player $i$). 

For any location $\omega$, let $v_{i}^{m}\left(\omega\right)\in\mathbb{R}^{+}$
denote the reward for player $i$ when $m$ players, including player
$i$, find the prize in $\omega$. The reward for finding the prize
alone, $v_{i}^{1}\left(\omega\right)$, is also called the private
value of player $i$ (at location $\omega$). We assume that the finder's
reward is weakly decreasing in the number of joint finders (i.e.,
$v_{i}^{m+1}\left(\omega\right)\leq v_{i}^{m}\left(\omega\right)$
for any $m$ and $\omega$), which reflects the negative impact of
search duplication. An example of such decreasing rewards is $v_{i}^{m}\left(\omega\right)=\frac{1}{m}\cdot v_{i}^{1}\left(\omega\right)$,
which may correspond to a setup in which one of the players who search
in the prize\textquoteright s location is randomly chosen to be its
undisputed owner, and she gains the prize\textquoteright s full value
(see, e.g., \citealp{fershtman1997simple}). 

In addition to the players, we introduce an external entity, society,
who is not one of the players and is indifferent to the identity of
the prize finder, as long as the prize is found. In our normative
analysis we set the objective of maximizing society's payoff. One
can think of society as representing a government who cares for the
welfare of those in society (e.g., consumers or patients) who will
be affected by the discovery. For any location $\omega$, let $v_{\mathfrak{s}}\left(\omega\right)\in\mathbb{R}^{+}$
denote the prize's social value for society when the prize is found
in $\omega$. Note that the social value does not depend on the identity
or the number of the prize's finders. In particular, the social value
of the prize is not reduced when there are multiple finders, which
seems plausible in various setups. For example, it seems plausible
that price competition between competing pharmaceutical firms will
not harm society (it might even benefit the consumers), and that the
social gain from a new discovery is not likely to be reduced when
two scientists fight over the credit.

Our main results assume that society disregards the players' search
costs. This assumption seems reasonable in setups where the potential
social impact of a discovery overshadows (in society's eyes) the player's
individual gains and costs, as in the motivating example of finding
a vaccine. In other setups this assumption might be less appropriate,
and we extend our result to a setup in which society internalizes
the players' search costs in Section \ref{subsec:Cost-Inclusive-Social-Payoff}.

We say that the game has common values if $v_{i}^{1}\left(\omega\right)=v_{j}^{1}\left(\omega\right)=v_{\mathfrak{s}}\left(\omega\right)$
for every two players $i,j\in N$ and every location $\omega\in\Omega$.
Summarizing all the above components allows us to define a search
game as a tuple $G=\left(N,\Omega,\Pi,\mu,K,c,v\right)$. 

\paragraph{Private Payoffs and Equilibrium}

Fix a strategy profile $s\in S$. Let $m_{s}\left(\omega\right)$
denote the number of players who search in $\omega$ when the prize's
location is $\omega$, i.e., 
\[
m_{s}\left(\omega\right)=\sum_{i\in N}\boldsymbol{1}_{\omega\in s_{i}\left(\pi_{i}\left(\omega\right)\right)}.
\]

The reward (resp., cost) of player $i$ conditional on the prize's
location being $\omega$ is equal to $1_{\omega\in s_{i}\left(\pi_{i}\left(\omega\right)\right)\,}v_{i}^{m_{s}\left(\omega\right)}\left(\omega\right)$
(resp., $c_{i}\left(\left|s_{i}\left(\pi_{i}\left(\omega\right)\right)\right|\right)$).
Thus, the (net) payoff of player $i$ conditional on the prize's location
being $\omega$, denoted by $u_{i}\left(s|\omega\right)$, is 
\[
u_{i}\left(s|\omega\right)=1_{\omega\in s_{i}\left(\pi_{i}\left(\omega\right)\right)\,}v_{i}^{m_{s}\left(\omega\right)}\left(\omega\right)-c_{i}\left(\left|s_{i}\left(\pi_{i}\left(\omega\right)\right)\right|\right).
\]
The players and society are both risk neutral with respect to their
payoffs. The (ex-ante) expected (net) payoff of player $i$ is given
by $u_{i}\left(s\right)=\sum_{\omega\in\Omega}\,\mu\left(\omega\right)\cdot u_{i}\left(s|\omega\right).$

A strategy profile $s=\left(\sSubOne,...,s_{n}\right)$ is a (Bayesian)
Nash equilibrium of search game $G$ if no player can gain by unilaterally
deviating from the equilibrium; i.e., if for every player $i$ and
every strategy $s'_{i}$ the following inequality holds: $u_{i}\left(s\right)\geq u_{i}\left(s'_{i},s_{-i}\right),$
where $s_{-i}$ describes the strategy profile played by all players
except player $i$.

\paragraph{Social Payoff}

Fix a strategy profile $s\in S$. Let $U\left(s|\omega\right)=v_{\mathfrak{s}}\left(\omega\right)\cdot\boldsymbol{1}_{m_{s}\left(\omega\right)\geq1}$
denote the social payoff, conditional on the prize's location being
$\omega$. The expected social payoff is equal to $U\left(s\right)=\sum_{\omega\in\Omega}\,\mu\left(\omega\right)\cdot U\left(s|\omega\right)$.
Let $U_{\textrm{opt}}$ denote the socially optimal payoff (or the
first-best payoff): $U_{\textrm{opt}}=\max_{s\in S}U\left(s\right).$
A strategy profile $s$ is socially optimal if it achieves the socially
optimal payoff, i.e., if $U\left(s\right)=U_{\textrm{opt}}$. 

A strategy profile is location-maximizing if it maximizes the number
of locations in which the prize is found; i.e., if for any strategy
profile $s'\in S$, 
\[
\sum_{\omega\in\Omega}\boldsymbol{1}_{\left\{ m_{s}\left(\omega\right)\geq1\right\} }\geq\sum_{\omega\in\Omega}\boldsymbol{1}_{\left\{ m_{s'}\left(\omega\right)\geq1\right\} .}
\]
The set of socially optimal strategy profiles is typically different
from the set of location-maximizing strategy profiles. The two notions
coincide if society assigns the same expected value to every location,
i.e., if $\mu\left(\omega\right)v_{\mathfrak{s}}\left(\omega\right)=\mu\left(\omega'\right)v_{\mathfrak{s}}\left(\omega'\right)$
for any two locations $\omega,\omega'\in\Omega$. A strategy profile
is exhaustive if the prize is always found, i.e., if $m_{s}\left(\omega\right)\geq1$
for every $\omega\in\Omega$. It is immediate that an exhaustive strategy
profile is both socially optimal and location-maximizing.

\section{Socially Optimal Equilibrium\label{sec:Socially-Optimal-Equilibrium}}

In this section we present conditions under which the strategic constraints
(namely, each player maximizing her private payoff) do not limit the
social payoff; that is, we give sufficient conditions for the existence
of socially optimal equilibria. 

\subsection{Search Games are Weakly Acyclic}

A sequence of strategy profiles is an improvement path (\citealp{monderer1996potential})
if each strategy profile differs from its preceding profile by the
strategy of a single player, who obtained a lower payoff in the preceding
profile.
\begin{defn}
\label{def:unliateral-improvements}A sequence of strategy profiles
$\left(\sOne,...,\sTee\right)$ is an improvement path if for every
$t\in\left\{ 1,...,T-1\right\} $ there exists a player $i_{t}\in N$
such that: (1) $s_{j}^{t}=s_{j}^{t+1}$ for every player $j\neq i_{t}$,
and (2) $u_{i_{t}}\left(s^{t+1}\right)>u_{i_{t}}\left(s^{t}\right)$.
\end{defn}
We begin by presenting an auxiliary result, which states that search
games are weakly acyclic: starting from any strategy profile, there
exists an improvement path that ends in a Nash equilibrium.\footnote{The proof introduces an agent-normal form representation of our game
(in the spirit of \citealp{selten1975reexamination}), which is similar
to matroid congestion games with player-specific payoffs. \citet[Theorem 8]{ackermann2009pure}
show that these games are weakly acyclic. Their result cannot be directly
applied to our setup, as there are some technical differences; most
notably, our cost function being non-linear (while in \citeauthor{ackermann2009pure}'s
setup the cost of searching in two locations must be the sum of the
costs in each location). Nevertheless, the proofs turn out to be similar.}
\begin{defn}[\citealp{milchtaich1996congestion}]
A game is weakly acyclic if for any $\sOne\in S$, there exists an
improvement path $\left(\sOne,...,\sTee\right)$, such that $\sTee$
is a (pure) Nash equilibrium.
\end{defn}
\begin{prop}
\label{pro:sequence-unilateral-improvements}Any search game is weakly
acyclic.
\end{prop}
\begin{proof}[Sketch of proof; formal proof is in Appendix \ref{subsec:Proof-of-weakly-acyclic}]
Define the payoff of a cell $\pi_{i}\in\Pi_{i}$ as the expected
payoff of player $i$ given that her signal is $\pi_{i}.$ Note that
player $i$ is best-responding iff every cell of $i$ is best-responding.
Player $i$ has $K_{i}$ units of capacity, which we index by $j=1,\ldots,K_{i}$.
A cell-unit of player $i$ is a pair $(\pi_{i},j)$, where $\pi_{i}\in\Pi_{i}$
is a cell, and $j$ a unit index. W.l.o.g. we assume that a strategy
chooses a specific location for every cell-unit $\alpha$, or chooses
that $\alpha$ be inactive. We define the payoff of a cell-unit $\left(\pi_{i},j\right)$
as the payoff of the cell $\pi_{i}$. Note that this payoff equals
the sum of the (interim) expected rewards in the locations of $\pi_{i}$'s
active cell-units, minus the cost of activating that many cell-units.

Given a strategy profile, suppose that there is no single inactive
cell-unit whose activation improves its own (i.e., the cell's) payoff.
Then activating multiple cell-units does not improve the cell's payoff
either, because of the convexity of the cost function. The case of
deactivation is similar. Therefore, we can show that a cell $\pi_{i}$
is best-responding iff every cell-unit of $\pi_{i}$ is best-responding.

The key part is Lemma \ref{lem:cell-unit} that says that if the members
of a set $B$ of cell-units (of various players) are best-responding,
and $\alpha\notin B$ is another cell-unit, then there is a sequence
of cell-unit improvements that ends with all the members of $B\cup\{\alpha\}$
best-responding. To prove weak acyclicity, start from any profile
$\sOne$, and using this lemma inductively add one cell-unit at a
time, until eventually everyone is best-responding.

To prove the lemma, we construct a sequence of improvements by the
members of $B\cup\{\alpha\}$. First, let $\alpha$ switch from its
current choice to its best-response. If $\alpha$ was active before
the switch, we add a dummy player in the location $\omega^{1}$ that
$\alpha$ left. Now begins a sequence we call Phase I. Suppose that
$\alpha$ switched to some location $\omega^{2}$. While cell-units
(of $B\cup\{\alpha\}$) not located in $\omega^{2}$ are still best-responding,
those in $\omega^{2}$ may now prefer to switch because of the extra
cell-unit in $\omega^{2}$ (call $\omega^{2}$ the current ``plus
location''). Let one of them switch to its best-response $\omega^{3}$,
and then another cell-unit may switch from $\omega^{3}$, etc. Phase
I goes on until everyone is best-responding, unless someone switches
to $\omega^{1}$, in which case Phase I is immediately terminated.

If a cell-unit is deactivated on stage $t$, it will not incentivize
another cell-unit to deactivate on stage $t+1$, because of the convexity
of costs. Moreover, Phase I will end after stage $t$, since there
would not be any plus location. 

To see that Phase I cannot go on forever, consider a cell-unit $\beta$
that switches from location $\omega$ to location $\omega'$, making
$\omega'$ the new plus location. The switch must strictly increase
$\beta$'s expected reward, and later the expected reward in $\omega'$
cannot drop below its current level; it may only be higher (if the
plus is somewhere else). Thus, $\beta$'s expected reward will never
drop back to the level it was at before the switch, even if $\beta$
does not improve again. Therefore, Phase I cannot enter a cycle; hence,
it must end.

Let $\sigma^{*}$ denote the strategy profile when Phase I ends. At
this point we remove the dummy from $\omega^{1}$, and denote the
resulting profile by $s^{*}$. If Phase I ended because someone switched
to $\omega^{1}$, then everyone is best-responding under $s^{*}$,
and we are done. Otherwise, Phase I ended because everyone was best-responding
under $\sigma^{*}$, and now follows Phase II.

While Phase I can be described as restabilizing after one cell-unit
is added, the analogous Phase II restabilizes after one cell-unit
is removed. First, one cell-unit switches from some location $\omega'$
to the current ``minus location'' $\omega^{1}$, then another switches
to $\omega'$, etc. On each stage we choose a cell-unit switch that
is best for its cell, i.e., there exists no cell-unit switch that
yields a higher increase in that cell's payoff. 

Phase II must eventually end, by the argument analogous to that of
Phase I. Then everyone is best-responding, and the lemma is proven.
\end{proof}
In particular, Proposition \ref{pro:sequence-unilateral-improvements}
implies that: 
\begin{cor}
Any search game admits a pure Nash equilibrium.
\end{cor}

\subsection{Existence of a Socially Optimal Equilibrium }

We begin by defining two properties required for our first main result
(Theorem \ref{thm:balanced}).

\paragraph{Ordinal Consistency}

Our first property requires that the ordinal ranking of any player
over her expected private values within a cell is (weakly) compatible
with society\textquoteright s ranking. That is, we say that a search
game has ordinally consistent payoffs if for any two locations $\omega$
and $\omega'$ in the same cell of player $i$, if the expected private
value of player $i$ is strictly lower in $\omega$ than in $\omega'$,
then the expected social value is weakly lower in $\omega$. 
\begin{defn}
\label{def:consistent}Search game $G$ has ordinarily consistent
payoffs if for any player $i$, any cell $\pi_{i}\in\Pi_{i}$, and
any two locations $\omega,\omega'\in\pi_{i}$, the following implication
holds:
\[
\mu\left(\omega\right)\cdot v_{i}^{1}\left(\omega\right)<\mu\left(\omega'\right)\cdot v_{i}^{1}\left(\omega'\right)\,\Rightarrow\,\mu\left(\omega\right)\cdot v_{\mathfrak{s}}\left(\omega\right)\leq\mu\left(\omega'\right)\cdot v_{\mathfrak{s}}\left(\omega'\right).
\]
Observe that having common values implies that the search game has
ordinally consistent payoffs. Further observe that if society has
uniform expected values (i.e., if $\mu\left(\omega\right)\cdot v_{\mathfrak{s}}\left(\omega\right)=\mu\left(\omega'\right)\cdot v_{\mathfrak{s}}\left(\omega'\right)$
for any two locations $\omega,\omega'\in\Omega$), then the search
game has ordinally consistent payoffs regardless of what the players'
private payoffs are.
\end{defn}

\paragraph{Solitary-Search Dominance}

Solitary-search dominance requires that any player always prefer searching
alone in any location to (1) searching jointly with other players
in another location within the same cell, or (2) leaving some of her
search capacity unused. Formally:
\begin{defn}
\label{def:balanced}Search game $G$ has solitary-search dominant
payoffs if
\begin{equation}
\mu\left(\omega|\pi_{i}\right)\cdot v_{i}^{1}\left(\omega\right)\geq\mu\left(\omega'|\pi_{i}\right)\cdot v_{i}^{2}\left(\omega'\right),\textrm{ and}\label{eq:SSD1}
\end{equation}
\begin{equation}
\mu\left(\omega|\pi_{i}\right)\cdot v_{i}^{1}\left(\omega\right)\geq c_{i}\left(K_{i}\right)-c_{i}\left(K_{i}-1\right),\label{eq:SSD2}
\end{equation}
for any player $i$, any cell $\pi_{i}\in\Pi_{i}$, and any pair\footnote{Due to the assumption of the cost function being convex, (\ref{eq:SSD2})
implies that $\mu\left(\omega|\pi_{i}\right)\cdot v_{i}^{1}\left(\omega\right)\geq c_{i}\left(k\right)-c_{i}\left(k-1\right)$
for any $1\leq k\leq K_{i}.$ A similar assumption of the search cost
being sufficiently small so that players always prefer searching alone
to not using their search capacity appears in \citet{chatterjee2004rivals}. } $\omega,\omega'\in\pi_{i}$. 
\end{defn}
Suppose first that there is no asymmetric information (which corresponds
to the case where all players have trivial information in our model).
If either ordinal consistency or solitary-search dominance are not
assumed, it is relatively easy to construct a game that does not admit
a socially optimal equilibrium (we construct such examples on Section
\ref{subsec:Necessity-of-All}). It is also not hard to show, on the
other hand, that ordinal consistency and solitary-search dominance
imply the existence of a socially optimal equilibrium. By virtue of
Proposition \ref{pro:sequence-unilateral-improvements}, we can now
show that this remains true with asymmetric information, as these
two conditions are sufficient for any search game.
\begin{thm}
\label{thm:balanced}Let $G$ be a search game with ordinally consistent
and solitary-search dominant payoffs. Then there exists a socially
optimal (pure) equilibrium.
\end{thm}
\begin{proof}
Consider a pure strategy profile that maximizes the social payoff.
Proposition \ref{pro:sequence-unilateral-improvements} implies that
there is a finite sequence of unilateral improvements that ends in
a Nash equilibrium. In what follows we show that the properties of
ordinal consistency and solitary-search dominance jointly imply that
the social payoff cannot decrease along that sequence. Without loss
of generality we can assume that each unilateral improvement consists
of changing merely a single choice within a single cell, since this
is in fact what the proof of Proposition \ref{pro:sequence-unilateral-improvements}
shows.

First we note that in each improvement, if the improving player leaves
a location in which there were multiple searchers, then the social
payoff cannot decrease. Next, solitary-search dominance implies that
if she leaves a location in which she is the sole searcher, then she
moves to an unoccupied location, as moving to an occupied location
would contradict (\ref{eq:SSD1}), and ``quitting'' (namely, deactivating
that unit of capacity) would contradict (\ref{eq:SSD2}). Finally,
ordinal consistency implies that if she moves to being the sole searcher
in another location, then the social payoff must weakly increase.
\end{proof}
In the socially optimal equilibrium, search costs may sometimes deter
a player from searching in some location $\omega$ if other players
might search there as well. Inequality (\ref{eq:SSD2}) merely states
that she will never be deterred by costs if she can search in $\omega$
alone.

Our next result states that even without the ordinal consistency assumption,
some efficiency is still guaranteed, in the sense that there exists
an equilibrium that maximizes the number of locations in which the
players search. Formally:
\begin{cor}
\label{cor:location-maximizing}Every search game $G$ with solitary-search
dominant payoffs admits a location-maximizing equilibrium.
\end{cor}
\begin{proof}
Let $\hat{G}=\left(N,\Omega,\Pi,\mu,K,c,\hat{v}\right)$ be a search
game similar to $G=\left(N,\Omega,\Pi,\mu,K,c,v\right)$, except that
$\hat{v}_{\mathfrak{s}}\left(\omega\right)=\nicefrac{1}{\mu\left(\omega\right)}$
for any $\omega\in\Omega$. Observe that $\hat{G}$ has ordinally
consistent and solitary-search dominant payoffs. This implies that
$\hat{G}$ admits a socially optimal equilibrium $\hat{s}$. Observe
that the definition of $\hat{v}_{\mathfrak{s}}$ implies that $\hat{s}$
is a location-maximizing strategy profile. Further observe that $\hat{s}$
is also an equilibrium of $G$ (as $G$ and $\hat{G}$ differ only
in the social payoff).
\end{proof}
In particular, any game with solitary-search dominant payoffs that
admits an exhaustive strategy profile, also admits an exhaustive equilibrium.

\paragraph{Price of Stability/Anarchy}

Theorem \ref{thm:balanced} states that there is an equilibrium that
maximizes the social payoff (i.e., that the price of stability is
1)\footnote{The price of stability (resp., anarchy) is defined as the ratio between
the socially optimal payoff $U_{\textrm{opt}}$ and the maximal (resp.,
minimal) social payoff induced by a Nash equilibrium; i.e., $PoS=\frac{U_{\textrm{opt}}}{\max_{s\in NE\left(G\right)}U\left(s\right)}$
and $PoA=\frac{U_{\textrm{opt}}}{\min_{s\in NE\left(G\right)}U\left(s\right)}$,
where $NE\left(G\right)$ is the set of Nash equilibria.} in any search game with ordinally consistent and solitary-search
dominant payoffs. By contrast, Figure \ref{fig:example-price-1} demonstrates
that the social payoff might be substantially lower in other Nash
equilibria (i.e., a price of anarchy larger than 1). \textcolor{brown}{Thus,
Theorem \ref{thm:balanced} is arguably more compelling in environments
where society is able to induce the play of the socially optimal equilibrium,
rather than other equilibria.}

\begin{figure}[h]
\begin{centering}
\includegraphics[scale=0.36]{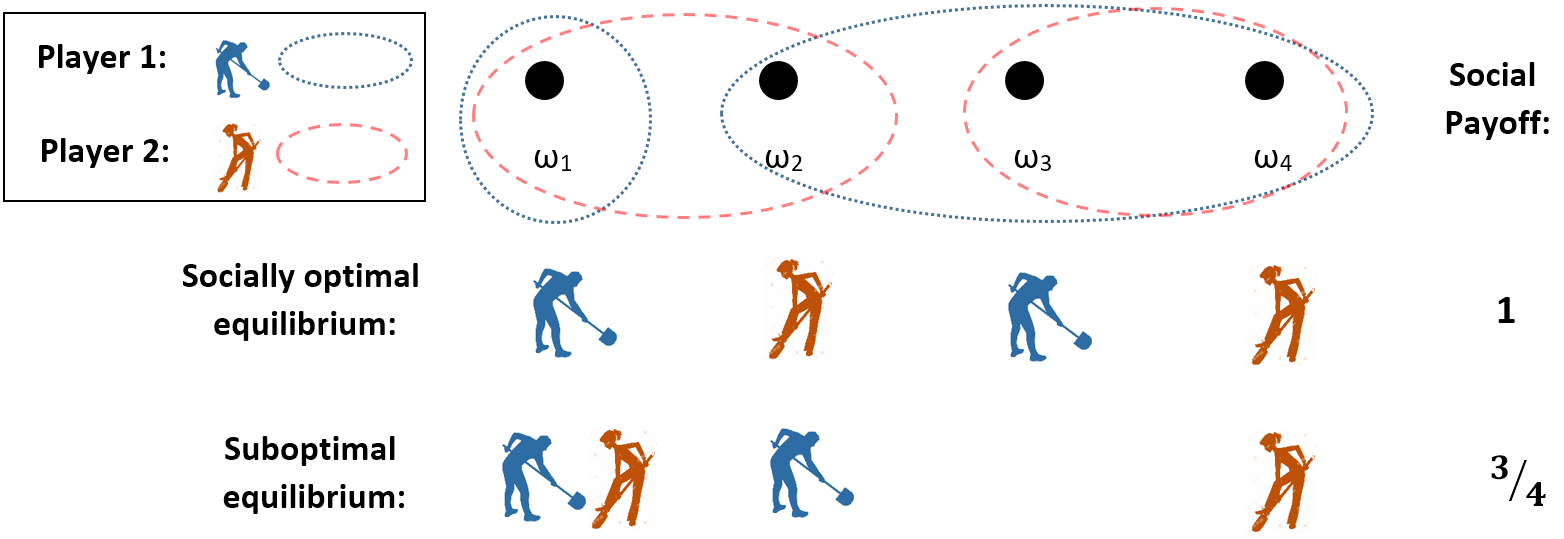}
\par\end{centering}
\caption{\label{fig:example-price-1}\textbf{Example for the price of anarchy}.
The figure presents two equilibria in a two-player search game with
ordinally consistent and solitary-search dominant payoffs (the ellipses
represent the partition elements), uniform prior, costless search
($c\equiv0$), reward of $v_{i}^{m}\equiv\frac{1}{m}$, social value
$v_{\mathfrak{s}}\equiv1$, and a capacity of one for every player.
The figure shows the location searched by each player for each possible
signal. For example, in the socially optimal equilibrium Player 1
searches in location $\omega_{1}$ when observing the signal $\left\{ \omega_{1}\right\} $
and searches in $\omega_{3}$ when observing the signal $\left\{ \omega_{2},\omega_{3},\omega_{4}\right\} $.
The first (resp., second) equilibrium is (resp., is not) socially
optimal with a social payoff of 1 (resp., 0.75). }
\end{figure}

\subsection{Necessity of All Assumptions in Theorem \ref{thm:balanced}\label{subsec:Necessity-of-All}}

The following three examples demonstrate that all the assumptions
of Theorem \ref{thm:balanced} are necessary to guarantee the existence
of a socially optimal equilibrium.

\paragraph{Necessity of Solitary-Search Dominance}

Example \ref{exa:neccesary-balnaced} demonstrates that solitary-search
dominance is necessary for Theorem \ref{thm:balanced}.
\begin{example}
\label{exa:neccesary-balnaced} For any $r\in\left(0,1\right)$ let
\[
G=\left(N=\left\{ 1,2\right\} ,\Omega=\left\{ \omega,\omega'\right\} ,\Pi\equiv\left\{ \Omega\right\} ,\mu,K\equiv1,c\equiv0,\left(v_{i}^{1}\equiv1,v_{i}^{2}\equiv r,v_{\mathfrak{s}}\equiv1\right)\right)
\]
be a two-player search game with trivial information partitions (namely,
each partition $\Pi_{i}$ contains a single element, which is the
whole $\Omega$), and a common prior $\mu$ defined as follows: $\mu\left(\omega\right)=\frac{2}{3}$
and $\mu\left(\omega'\right)=\frac{1}{3}$. Both locations induce
a private value of 1 to a sole searcher and a private value of $r\in\left(0,1\right)$
in case of simultaneous searches. Note that $G$ has ordinally consistent
payoffs, and that it satisfies solitary-search dominance iff $r\leq0.5$.
In what follows we show that for any $r>0.5$ the unique best-reply
against an opponent who searches in location $\omega$ is to search
in $\omega$ as well (which implies that searching in $\omega$ is
a dominant strategy). This is so because searching in $\omega$ yields
an expected payoff of $\frac{2}{3}\cdot r$, while searching in $\omega'$
yields $\frac{1}{3}\cdot1$. This, in turn, implies that the unique
equilibrium is both players searching in $\omega$, which is suboptimal.
\end{example}

\paragraph{Necessity of Ordinal Consistency}

Example \ref{exa:neccesary-consistent} demonstrates that the consistency
requirement is necessary to guarantee the existence of a socially
optimal equilibrium. Specifically, it shows that even for one-player
search games, and even when society and the player have the same ordinal
ranking over the values of the prize in each location and search is
costless, the unique Nash equilibrium is not necessarily socially
optimal if the ordinal consistency requirement is not satisfied.
\begin{example}
\label{exa:neccesary-consistent} Let $G=\left(N=\left\{ 1\right\} ,\Omega=\left\{ \omega,\omega'\right\} ,\Pi_{1}=\left\{ \Omega\right\} ,\mu,K_{1}=1,c_{1}=0,v\right)$
be a one-player search game with a prior $\mu\left(\omega\right)=\nicefrac{1}{4}$,
$\mu\left(\omega'\right)=\nicefrac{3}{4}$, and with values of $v_{\mathfrak{s}}\left(\omega\right)=2$,
$v_{\mathfrak{s}}\left(\omega'\right)=1$, $v_{_{1}}^{{\scriptscriptstyle 1}}\left(\omega\right)=4$,
and $v_{_{1}}^{{\scriptscriptstyle 1}}\left(\omega'\right)=1$. Observe
that the game's payoffs are trivially solitary-search dominant due
to having a single player and a costless search. It is simple to see
that the player searches in location $\omega$ in the unique equilibrium,
although this yields a lower social payoff than searching in $\omega'$.
\end{example}

\paragraph{Necessity of Simultaneous Searches }

An (implicit) key assumption in our model is that all searches are
done simultaneously. The following example demonstrates that Theorem
\ref{thm:balanced} is no longer true if players search sequentially. 
\begin{example}[\label{exam:seq}Sequential play, see Figure \ref{fig:Illustration-of-Example-seq}]
Let

\begin{figure}[h]

\caption{\label{fig:Illustration-of-Example-seq}Illustration of Example \ref{exam:seq}:
Sequential Play}

\begin{centering}
\includegraphics[scale=0.37]{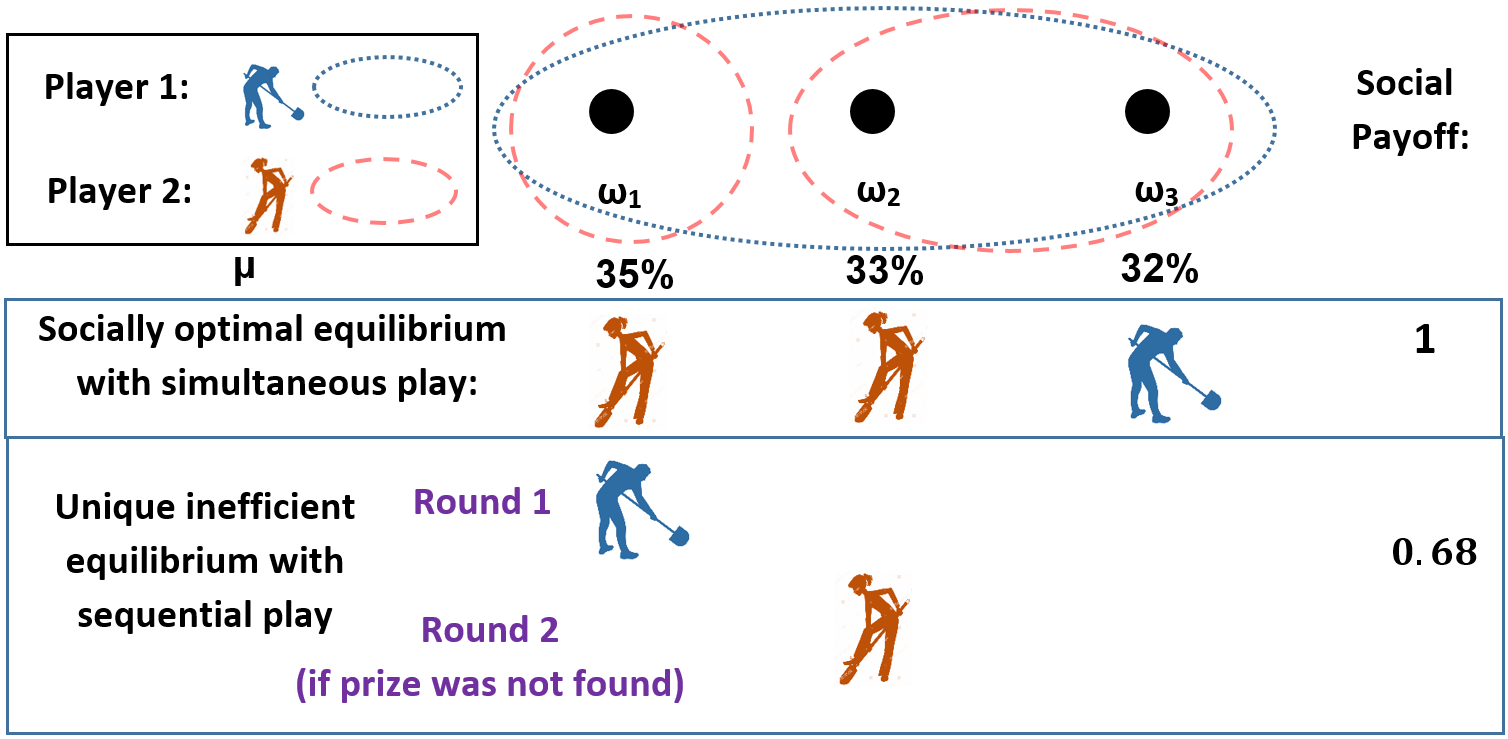}
\par\end{centering}
\end{figure}
\[
G=\left(N=\left\{ 1,2\right\} ,\Omega=\left\{ \omega_{1},\omega_{2},\omega_{3}\right\} ,\Pi,\mu,K\equiv1,c\equiv0,\left(v_{i}^{1}\equiv1,v_{i}^{2}\equiv0.5,v_{\mathfrak{s}}\equiv1\right)\right)
\]
 be a two-player search game with three locations, capacity 1 for
each player and costless search. All locations yield a private value
of 1, which is equally shared between simultaneous finders. Player
1 has the trivial partition $\Pi_{1}=\left\{ \Omega\right\} $, while
Player $2$ knows if the prize is in location $\omega_{1}$ or not
(i.e., $\Pi_{2}=\left\{ \left\{ \omega_{1}\right\} ,\left\{ \omega_{2},\omega_{3}\right\} \right\} $.
The  prior assigns slightly higher (resp., lower) probability to location
$\omega_{1}$ (resp., $\omega_{3}$), i.e., $\mu\left(\omega_{1}\right)=35\%$,
$\mu\left(\omega_{2}\right)=33\%$, $\mu\left(\omega_{3}\right)=32\%$.
Observe that the game satisfies ordinal consistency and solitary-search
dominance. In our model, in which players search simultaneously, the
game admits two (pure) Nash equilibria, both of which are socially
optimal: Player 1 searches in either location 2 or 3, and Player 2
searches in the remaining two locations. 

By contrast, if the game is sequential and Player 1 plays first, then
the game admits a unique equilibrium, which is not efficient: Player
1 searches in $\omega_{1}$, and if the prize has not been found,
Player 2 searches in location $\omega_{2}$ (and no player searches
in location $\omega_{3}$). Note that this profile is the unique equilibrium
regardless of whether or not the model lets Player 2 observe the location
in which Player 1 searched in the previous round.
\end{example}

\subsection{Implications for Innovation Contests\label{subsec:Insights-for-Innovation}}

Consider the setup of an innovation contest, in which a contest designer,
who wishes to maximize the social payoff, might influence the private
payoffs of players by offering a monetary bonus to the prize's finder,
which is added to the inherent reward. In what follows we sketch a
few implications of Theorem \ref{thm:balanced} in a contest with
asymmetric information, while leaving the interesting question of
characterizing the optimal bonuses in this setup to future research. 

Observe first that if the private payoffs satisfy ordinal consistency
and solitary-search dominance, then Theorem \ref{thm:balanced} implies
that the designer can maximize the social payoff without offering
any bonus: the designer is only required to be able to give nonenforced
recommendations to the players (which allows him to induce the play
of the socially optimal Nash equilibrium, rather than other equilibria).
In what follows we consider the case in which solitary-search dominance
is violated in the search game (without additional monetary bonuses).

Consider first a setup in which the contest designer can only offer
a constant bonus, which is independent of the prize's location. A
constant bonus can help to increase the relative expected private
value of locations with a high prior probability. As a result, it
can help obtain the optimal social payoff, when the reason for not
having the required properties without the designer's intervention
is a low-prior location having a too-high private value. For example,
consider a search game with costless search (i.e., $c\equiv0$), where
there are two locations $\omega,\omega'$ in the same cell of player
$i$ with priors $\mu\left(\omega\right)=0.1$ and $\mu\left(\omega'\right)=0.2$
and with private values of $v_{i}^{1}\left(\omega\right)=5$ and $v_{i}^{1}\left(\omega'\right)=1$,
and $v_{i}^{m}\equiv\frac{1}{m}v_{i}^{1}$. The too-high private value
of location $\omega$ violates solitary-search dominance because the
expected private value in $\omega$ ($0.5=0.1\cdot5$) is more than
twice the expected private value in $\omega'$ ($0.2\cdot1$). A constant
bonus of $1$ would restore solitary-search dominance (making the
expected private value of $\omega$ and $\omega'$ to be equal to
$0.6=0.1\cdot\left(5+1\right)$ and $0.4=0.2\cdot\left(1+1\right)$,
respectively).

When the designer can offer a location-dependent and player-dependent
bonus, it allows him to obtain solitary-search dominance and ordinal
consistency when faced with any profile of rewards. An interesting
open question is how the designer can maximize the social payoff,
while minimizing the expected bonus. For example, assume that the
payoffs are ordinally consistent, but they are not solitary-search
dominant. Theorem \ref{thm:balanced} suggests that the designer should
boost locations that have lower expected private values (which violate
solitary-search dominance). Note that these locations might not coincide
with the locations that are not searched by any player in the inefficient
equilibrium. This is demonstrated in Example \ref{exa:innocation-competition}.
\begin{example}
\label{exa:innocation-competition}Consider the following search game
with common values (as illustrated in Figure \ref{fig:Illustration-of-Example}):
\begin{figure}[h]
\begin{centering}
\includegraphics[scale=0.35]{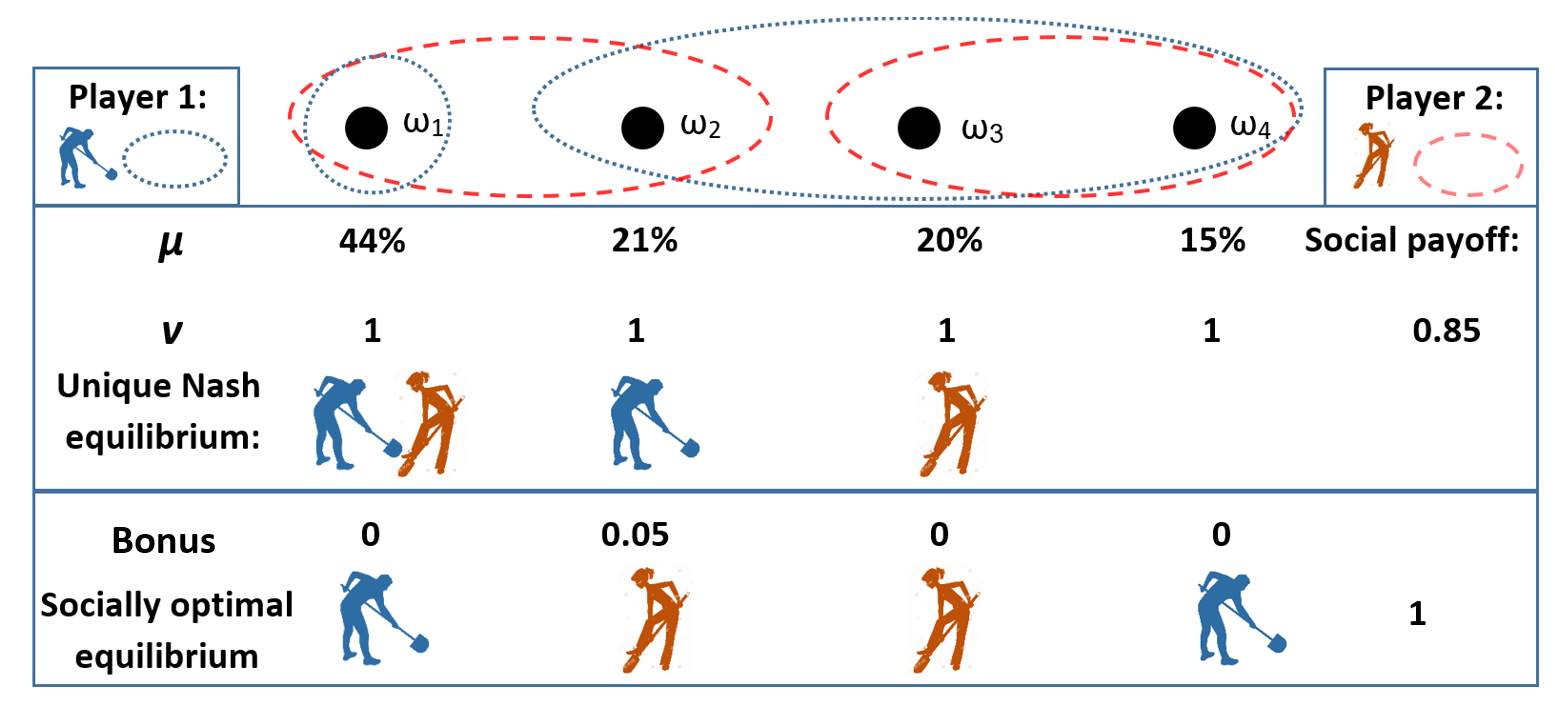}
\par\end{centering}
\caption{\label{fig:Illustration-of-Example}Illustration of Example \ref{exa:innocation-competition}:
Impact of Monetary Bonuses on the Social Payoff}
\end{figure}
$\left(N=\left\{ 1,2\right\} ,\Omega=\left\{ \omega_{1},\omega_{2},\omega_{3},\omega_{4}\right\} ,\Pi,\mu,K\equiv1,c\equiv0,v_{i}^{m}\equiv\frac{1}{m},v_{\mathfrak{s}}\equiv1\right)$,
where the prior is $\mu\left(\omega_{1}\right)=44\%$, $\mu\left(\omega_{2}\right)=21\%,$
$\mu\left(\omega_{3}\right)=20\%$ and $\mu\left(\omega_{4}\right)=15\%$,
player 1 observes whether the prize's location is 1 or not, i.e.,
$\Pi_{1}=\left\{ \left\{ \omega_{1}\right\} ,\left\{ \omega_{2},\omega_{3},\omega_{4}\right\} \right\} $,
and player 2 observes whether the prize's location is at most 2 or
not, i.e., $\Pi_{2}=\left\{ \left\{ \omega_{1},\omega_{2}\right\} ,\left\{ \omega_{3},\omega_{4}\right\} \right\} $.
The game admits a unique equilibrium, in which player 1 searches in
$\omega_{1}$ and $\omega_{2}$, while player 2 searches in locations
$\omega_{1}$ and $\omega_{3}$. This equilibrium yields an expected
social payoff of 0.85 because no player searches in $\omega_{4}$.
Note that solitary-search dominance is violated because of the low
probability of location $\omega_{2}$ (rather than a low probability
of $\omega_{4}$).

If the designer can offer a bonus of 0.05 that increases the private
value in location $\omega_{2}$ by 5\% to 1.05 (which requires a modest
expected bonus of $21\%\cdot0.05\approx0.01$), then the modified
rewards satisfy solitary-search dominance, and, as a result, the game
admits a socially optimal equilibrium with a social payoff of 1 (in
which player 1 searches in locations $\omega_{1}$ and $\omega_{4}$,
while player 2 searches in locations $\omega_{2}$ and $\omega_{3}$).
\end{example}

\section{Socially Optimal Payoff\label{sec:socially-optimal-payoff}}

Theorem \ref{thm:balanced} has provided conditions under which the
socially optimal (first-best) payoff is also yielded by some equilibrium.
In this section we present lower bounds (which are binding in many
cases, as demonstrated below) for the socially optimal payoff, namely,
for the highest social payoff yielded by any strategy profile. Thus,
we do not explicitly mention equilibria in this section; nevertheless,
we remind the reader that by Theorem \ref{thm:balanced} the socially
optimal payoff yielded in every result or example of this section
is also yielded by some equilibrium of the game, if payoffs are ordinally
consistent and solitary-search dominant.

Given a profile of strategies, society cares only about which locations
are being searched (by anyone) and which are not. Define a pure outcome
as a function $f:\Omega\rightarrow\left\{ 0,1\right\} $ that specifies
which locations are being searched. A strategy profile induces a pure
outcome, and, of course, not every outcome can be induced by strategies.
Similarly, we define a mixed outcome as a function $f:\Omega\rightarrow\left[0,1\right]$,
where $f\left(\omega\right)$ is the probability that $\omega$ is
being searched. A lottery over strategy profiles induces a mixed outcome.
The notion of mixed outcomes turns out to be helpful for deriving
and presenting our lower bounds for the socially optimal payoff, as
demonstrated in the following example.
\begin{example}[``Three coins'']
\label{exa:3-coins}Consider a game where the locations $\Omega=\left\{ 0,1\right\} ^{3}$
are vectors of three binary coordinates. Each player has a capacity
of one, the social payoff $v_{\mathfrak{s}}\equiv1$ is one in every
location, and $\mu$ is uniform over $\Omega$. There are three players,
and player $i$ knows the $i$-th coordinate of the prize's location.
Intuitively, this game can be interpreted as nature tossing three
coins to determine the prize's location, with each player observing
the result of one of these coin tosses. Let $f$ be the mixed outcome
that assigns $f\left(\omega\right)=3/4$ to every location $\omega$.
We will see that $f$ can be induced by a lottery over strategy profiles.
This implies that the socially optimal payoff is at least $3/4$.
\end{example}
Note that the socially optimal payoff in Example \ref{exa:3-coins}
cannot exceed $3/4$. The reason is that for any strategy profile,
the number of searched locations is at most the number of cells of
all players (multiplied by the capacity), $\sum_{i\in N}\thinspace K_{i}\cdot\left|\Pi_{i}\right|$\,,
which is six in the example, while $\left|\Omega\right|=8$. Similarly,
for any subset of locations $W\subseteq\Omega$, the number of searched
locations within $W$ is at most the number of cells that intersect
$W$ (namely, $\sum_{i\in N}K_{i}\cdot\left|\left\{ \pi_{i}\in\Pi_{i}\right\} :\pi_{i}\cap W\neq\emptyset\right|$).
Next, consider a mixed outcome $f$ induced by a lottery over strategy
profiles. The mixed outcome must satisfy that the expected number
of searched locations within $W$ (i.e., $\sum_{\omega\in W}f\left(\omega\right))$
is at most the number of cells that intersect $W$; that is, $f$
must satisfy the following compatibility property.
\begin{defn}
\label{def:compatibility}A mixed outcome $f$ is compatible with
the information structure (henceforth, compatible) if for any subset
of locations $W\subseteq\Omega$, the following inequality holds:
\begin{equation}
\sum_{\omega\in W}f\left(\omega\right)\leq\sum_{i\in N}K_{i}\cdot\left|\left\{ \pi_{i}\in\Pi_{i}\right\} :\pi_{i}\cap W\neq\emptyset\right|\,\,.\label{eq:compatibility}
\end{equation}
\end{defn}
We will later show that the opposite also holds; namely, any compatible
outcome can be induced by a lottery over profiles. This will allow
us to derive our first lower bound for the socially optimal payoff. 

Let us now verify that $f$ of Example \ref{exa:3-coins} is compatible.
For any $W\subseteq\Omega$, the LHS of (\ref{eq:compatibility})
equals $3/4\cdot\left|W\right|$\,. The partition of every player
$i$ consists of two cells, each of size four. Therefore, the number
of cells of player $i$ that intersect $W$ is at least $\left|W\right|/4$;
hence, the RHS of (\ref{eq:compatibility}) is at least $\sum_{i\in N}\left|W\right|/4=3/4\cdot\left|W\right|=\sum_{\omega\in W}f\left(\omega\right)$. 

We now define another class of outcomes, those that are generated
by a fractional allocation, which will be helpful in deriving the
second lower bound for the socially optimal payoff.
\begin{defn}
A fractional allocation $\alpha=\left(\alpha_{1},\ldots,\alpha_{n}\right)$
specifies a nonnegative number $\alpha_{i}\left(\pi_{i},\omega\right)$
for every cell $\pi_{i}$ of player $i$ and every location $\omega\in\pi_{i}$,
such that $\sum_{\omega\in\pi_{i}}\,\alpha_{i}\left(\pi_{i},\omega\right)\leq K_{i}$\,
for any cell $\pi_{i}$\,.\\
A fractional allocation generates the mixed outcome $f_{\alpha}$
that assigns to every location $\omega$ the sum of the allocations
to $\omega$ over all players (as long as this sum does not exceed
one). That is, $f_{\alpha}\left(\omega\right)=\min$$\left(1,\,\sum_{i\in N}\thinspace\alpha_{i}\left(\pi_{i}(\omega),\omega\right)\,\right)$.
\end{defn}
\begin{example}
\label{exa:uniform-cells}Suppose that there are $n$ players, each
with a capacity of one, and all cells are of the same size $m$, where
$m\geq n$. Let $\alpha_{i}$ divide each player's capacity equally
between the locations within each cell, i.e., $\alpha_{i}\left(\pi_{i},\omega\right)=1/m$\,
for each cell $\pi_{i}$ and location $\omega\in\pi_{i}$\,. Then
this fractional allocation $\alpha=\left(\alpha_{1},\ldots,\alpha_{n}\right)$
generates the outcome $f\left(\omega\right)=n/m$\, in every $\omega$.
\end{example}
Suppose that each player in Example \ref{exa:uniform-cells} independently
employs a randomized strategy that chooses, within each cell, each
location with equal probability $1/m$. This is equivalent to a lottery
over strategy profiles that gives equal probability to every pure
profile that exhibits no idleness (i.e., every profile in which all
players always use their entire search capacity). Search duplication
occurs under some of these profiles, which implies that the induced
outcome would be, in every $\omega$, strictly less than $n/m$. By
contrast, in our definition of an outcome generated by a fractional
allocation, the accumulation in every $\omega$ is done without any
``waste.'' Nevertheless, we now show that this wasteless outcome
can always be achieved by a lottery over profiles.

The following proposition states that both classes of outcomes defined
above coincide with the outcomes that can be induced by a lottery
over profiles. 
\begin{prop}
\label{prop:equivalent-outcome-classes}For a mixed outcome $f$,
the following conditions are equivalent:\\
(i) $f$ can be induced by a lottery over strategy profiles,\\
(ii) $f$ is compatible, and\\
(iii) $f$ can be generated by a fractional allocation.
\end{prop}
Proposition \ref{prop:equivalent-outcome-classes} implies that the
outcome $f\left(\omega\right)=3/4$ \,$\forall\omega$ in Example
\ref{exa:3-coins} is induced by some lottery over strategy profiles.
This outcome yields a social payoff of $3/4$ and, hence, the support
of the lottery must contain a profile that yields at least that much;
therefore, the socially optimal payoff, $U_{\textrm{opt}}$, is at
least $3/4$ (and hence $U_{\textrm{opt}}=3/4$). Similarly, in Example
\ref{exa:uniform-cells} the outcome $f\left(\omega\right)=n/m$ \,$\forall\omega$\,
is induced by some lottery over profiles; therefore, the socially
optimal payoff is at least\, $n/m\cdot\sum_{\omega\in\Omega\,}\mu\left(\omega\right)v_{\mathfrak{s}}\left(\omega\right)$.
Thus, in general, Proposition \ref{prop:equivalent-outcome-classes}
implies the following two useful lower bounds for the socially optimal
payoff.
\begin{thm}
\label{thm:lower-bounds}(i) For any compatible outcome $f$, the
socially optimal payoff is at least the social payoff yielded by $f$,
i.e., $U_{\textrm{opt}}\geq\sum_{\omega\in\Omega\,}f\left(\omega\right)\cdot\mu\left(\omega\right)v_{\mathfrak{s}}\left(\omega\right)$.\\
(ii) For any fractional allocation $\alpha$, the socially optimal
payoff is at least the social payoff yielded by the generated outcome,
i.e., $U_{\textrm{opt}}\geq\sum_{\omega\in\Omega\,}f_{\alpha}\left(\omega\right)\cdot\mu\left(\omega\right)v_{\mathfrak{s}}\left(\omega\right)$.
\end{thm}
In addition, Proposition \ref{prop:equivalent-outcome-classes} implies
the following characterization of the socially optimal payoff.
\begin{cor}
\label{cor:charcterization-1}$U_{\textrm{opt}}=\max_{f\in F_{C}}\,\sum_{\omega\in\Omega\,}f\left(\omega\right)\cdot\mu\left(\omega\right)v_{\mathfrak{s}}\left(\omega\right)=\max_{f\in F_{F}}\,\sum_{\omega\in\Omega\,}f\left(\omega\right)\cdot\mu\left(\omega\right)v_{\mathfrak{s}}\left(\omega\right)$,
where $F_{C}$ denotes the set of compatible outcomes and $F_{F}$
denotes the set of outcomes generated by fractional allocations.
\end{cor}
\begin{proof}[Sketch of proof of Proposition \ref{prop:equivalent-outcome-classes};
formal proof in Appendix \ref{subsec:Proof-of-Claim-prob-to-smooth}]
We already explained why (i)$\Rightarrow$(ii). To see that (ii)$\Rightarrow$(iii),
we construct a flow network: a directed graph whose edges have flow
capacities. The graph connects every cell to the locations contained
in it, with infinite flow capacity (as illustrated in Figure \ref{fig:min-cut-1}).
We add a source vertex that connects to every cell, with flow capacity
$K_{i}$, and a sink vertex to which every location $\omega$ is connected,
with flow capacity $f(\omega)$. A cut is a subset of edges without
which there exists no path from the source to the sink. The compatibility
of $f$ implies that the minimal cut has a total capacity of $\sum_{\omega\in\Omega}f\left(\omega\right).$
Therefore, by the max-flow min-cut theorem (\citealp{ford1956maximal};
see a textbook presentation in \citealp[p. 723, Thm. 26.6]{cormen2009introduction}),
the network admits a flow of $\sum_{\omega\in\Omega}f\left(\omega\right).$
We define a fractional allocation $\alpha$ that generates $f$ by
letting ${\color{brown}\alpha}_{i}\left(\pi_{i},\omega\right)$ equal
the flow from $\pi_{i}$ to $\omega$.
\begin{figure}[h]
\begin{centering}
\includegraphics[scale=0.35]{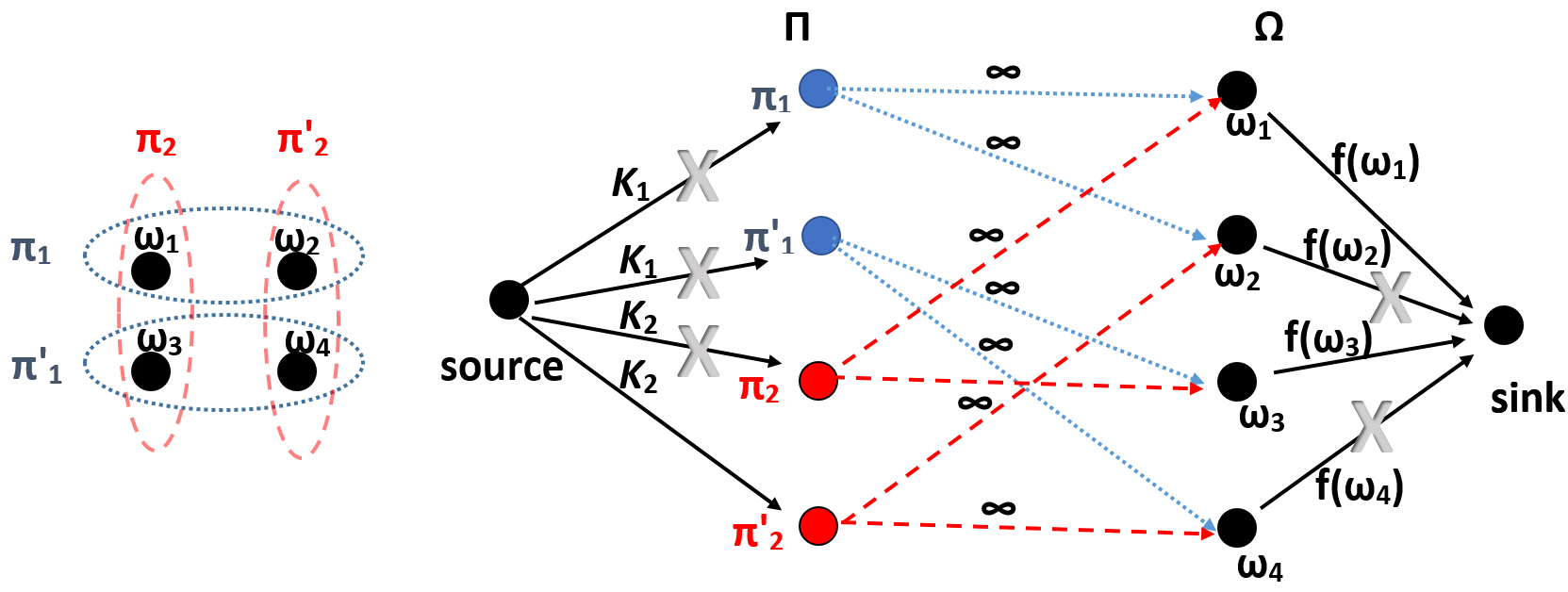}
\par\end{centering}
\caption{\label{fig:min-cut-1}\textbf{Illustration of }(ii)$\Rightarrow$(iii).
The left side of the figure demonstrates partitions in a two-player
search game. The right side demonstrates the constructed directed
graph in which (1) a source node is linked to every player's cells
by an edge with the player's capacity, and (2) each cell is linked
by unlimited edges to all the locations within that cell, and (3)
each location $\omega$ is linked to a sink node by an edge with capacity
$f\left(\omega\right)$. The gray X-s demonstrate an example of a
cut, i.e., a subset of edges whose removal from the graph disconnects
the source from the sink.}
\end{figure}

Now we show that (iii)$\Rightarrow$(i). To simplify the sketch of
the proof assume that all capacities are equal to one. Let $\alpha$
be a fractional allocation. We can represent $\alpha$ as a matrix
$\left(A_{\pi\omega}^{{\color{brown}\alpha}}\right)_{\omega\in\Omega,\pi\,\textrm{is\,a\,cell}}\,$,
where 
\[
A_{\pi\omega}^{{\color{brown}\alpha}}=\begin{cases}
{\color{brown}\alpha_{i}}\left(\pi,\omega\right) & \omega\in\pi,\,\pi\in\Pi_{i}\\
0 & \textrm{elsewhere\,.}
\end{cases}
\]
 Observe that $A_{\pi\omega}^{{\color{brown}\alpha}}$ is a nonnegative
matrix, and that the sum of each row $\pi$ is at most one. Let $B_{\pi\omega}^{{\color{brown}\alpha}}$
be a matrix derived from $A_{\pi\omega}^{{\color{brown}\alpha}}$
by decreasing elements of the matrix such that the sum of each column
$\omega$ that exceeded one in $A_{\pi\omega}^{{\color{brown}\alpha}}$
is equal to one in $B_{\pi\omega}^{{\color{brown}\alpha}}$. Observe
that $B_{\pi\omega}^{{\color{brown}\alpha}}$ is a doubly substochastic
matrix; i.e., it is a nonnegative matrix for which the sum of each
column and each row is at most one. A simple adaptation of the Birkhoff--von
Neumann theorem shows that $B_{\pi\omega}^{{\color{brown}\alpha}}$can
be represented as a convex combination of matrices $C_{\pi\omega}^{1},...,C_{\pi\omega}^{K}$
(i.e., $B_{\pi\omega}^{{\color{brown}\alpha}}=\sum w_{k}\cdot C_{\pi\omega}^{k}$
where $\sum w_{k}=1$ and $w_{k}\geq0$), where each matrix $C_{\pi\omega}^{k}$:
(1) contains only zeros and ones, and (2) contains in each row and
in each column at most a single value of one. Observe that each such
matrix $C_{\pi\omega}^{k}$ corresponds to a pure profile $s^{k}$
in the search game, and that the outcome $f_{{\color{brown}\alpha}}$
is a weighted sum of the outcomes induced by the profiles $s^{k}$.
This implies that $f_{{\color{brown}\alpha}}$ is induced by the lottery
over strategy profiles $\sigma=\sum w_{k}\cdot s^{k}$.
\end{proof}
Next, we present two examples that demonstrate the usefulness of the
lower bounds of Theorem \ref{thm:lower-bounds}. The first example
extends Example \ref{exa:uniform-cells}.

\addtocounter{example}{-1}
\begin{example}[revisited]
In Example \ref{exa:uniform-cells} we had $n$ players and cells
of uniform size $m\geq n$. In particular, if $m=n$ then we get that
$f\left(\omega\right)=n/m=1$ in every location. This implies that
the game admits an exhaustive strategy profile, i.e., $U_{\textrm{opt}}=\sum_{\omega\in\Omega\,}\mu\left(\omega\right)v_{\mathfrak{s}}\left(\omega\right)$.
Likewise, if $m<n$ then the game admits an exhaustive profile, since
$f\left(\omega\right)=\min\left(1,\,n/m\right)=1.$ 

More generally, when the capacity $K_{i}$ of each player $i$ is
not necessarily one and the cells are not of uniform size, let $M_{i}$
denote the size of the largest cell of player $i$. Then $f\left(\omega\right)\geq\min\left(1,\,\sum_{i\in N}\nicefrac{K_{i}}{M_{i}}\right)$
and $U_{\textrm{opt}}\geq\min\left(1,\,\sum_{i\in N}\nicefrac{K_{i}}{M_{i}}\right)\cdot\sum_{\omega\in\Omega\,}\mu\left(\omega\right)v_{\mathfrak{s}}\left(\omega\right)$.
In particular, if $\sum_{i\in N}\nicefrac{K_{i}}{M_{i}}\geq1$ then
the game admits an exhaustive strategy profile.\footnote{\label{fn:redundancy-free-condition}By a similar argument, let $m_{i}=\min\left(\left|\pi_{i}\right|:\pi_{i}\in\Pi_{i}\right)$;
then, if $\sum_{i\in N}\nicefrac{K_{i}}{m_{i}}\leq1$ then the game
admits a profile that always employs full capacity, while avoiding
search duplication (see also Definition \ref{def:redundancy-free}
in Section \ref{subsec:Cost-Inclusive-Social-Payoff}).}
\end{example}
We conclude this section with an example in which a compatible outcome
allows us to find the socially optimal payoff.
\begin{example}
Suppose that $\Omega$ contains ten vectors, each of four binary coordinates,
i.e., $\Omega\subset\left\{ 0,1\right\} ^{4}$ and $\left|\Omega\right|=10$.
There are four players, each of capacity one, and player $i$ knows
the $i$-th coordinate. Let us show that the outcome that assigns
$f\left(\omega\right)=0.8$ to every $\omega$ is compatible. Let
$W\subseteq\Omega$, and suppose first that all eight cells intersect
$W$. Then, since $\sum_{\omega\in W}\thinspace f\left(\omega\right)\leq\sum_{\omega\in\Omega}\thinspace f\left(\omega\right)=8,$
inequality (\ref{eq:compatibility}) holds. Now, a cell of player
$i$ does not intersect $W$ if every $\omega\in W$ resides in the
other cell of player $i$; i.e., the $i$-th coordinate is fixed across
$W$. Suppose that $k$ cells in all do not intersect $W$, $1\leq k\leq4$.
Then there are only $4-k$ ``free'' coordinates, therefore, $\left|W\right|\leq2^{4-k}$,
and hence $\sum_{\omega\in W}\thinspace f\left(\omega\right)\leq0.8\cdot2^{4-k}$.
Let us verify that this is less than the number of cells that do intersect
$W$, i.e., that\, $0.8\cdot2^{4-k}\leq8-k$. Indeed, for $k=1$
the LHS is $6.4$ and the RHS is $7$, and for larger $k$ the difference
increases. 

Since $f$ is compatible, if the social payoff is one in every location
and $\mu$ is uniform, then $U_{\textrm{opt}}=0.8$. Put differently,
there are eight cells (and thus at most eight locations can be searched);
and the compatibility of $f$ implies the existence of a strategy
profile under which eight locations are searched. 

Note that if $\Omega$ happens to have exactly five zeros and five
ones on each coordinate, then the fractional allocation that divides
equally within every cell generates the above outcome $f$. If $\Omega$
does not have this property, then identifying a fractional allocation
that generates $f$ might be somewhat difficult; nevertheless, Proposition
\ref{prop:equivalent-outcome-classes} tells us that such an allocation
does exist.
\end{example}

\section{Cost-Inclusive Social Payoff\label{subsec:Cost-Inclusive-Social-Payoff}}

In this subsection we consider a variant of our model, in which society
internalizes the players' search costs. We examine conditions under
which the game admits a socially optimal equilibrium, as in Theorem
\ref{thm:balanced}.

We define the cost-inclusive social payoff, conditional on the prize's
location being $\omega$, as 
\[
U_{CI}\left(s|\omega\right)=v_{\mathfrak{s}}\left(\omega\right)\cdot\boldsymbol{1}_{m_{s}\left(\omega\right)\geq1}-\sum_{i\in N}c_{i}\left(\left|s_{i}\left(\pi_{i}\left(\omega\right)\right)\right|\right),
\]
and the expected cost-inclusive social payoff is $U_{CI}\left(s\right)=\sum_{\omega\in\Omega}\,\mu\left(\omega\right)\cdot U_{CI}\left(s|\omega\right)$.

First, note that in the baseline model, in which society disregards
the players' costs, society always prefers player $i$ searching an
unoccupied location to player $i$ not searching. But here, with the
cost-inclusive social payoff, this may no longer be true. Thus, the
preference of an individual player $i$ who wishes to search despite
her cost (by the second part of the solitary-search dominance definition)
may diverge from society's preference, as society may wish player
$i$ to be idle. Hence, we add here the following assumption, which
says that the expected social value always exceeds any expected marginal
cost.
\begin{defn}
We say that social value outweighs cost if
\[
\mu\left(\omega\right)\cdot v_{\mathfrak{s}}\left(\omega\right)\geq\mu\left(\pi_{i}\right)\cdot\left(c_{i}\left(K_{i}\right)-c_{i}\left(K_{i}-1\right)\right)
\]
for any location $\omega$, player $i$, and cell $\pi_{i}$.
\end{defn}
Nevertheless, even under the assumption that social value outweighs
cost, on top of the assumptions of ordinal consistency and solitary-search
dominance in Theorem \ref{thm:balanced}, the game still need not
admit a socially optimal equilibrium, as the following simple example
demonstrates.
\begin{example}
Let $G=\left(N=\left\{ 1,2\right\} ,\Omega=\left\{ \omega\right\} ,\Pi,\mu,K\equiv1,c\equiv0.1,v_{i}^{m}\equiv\frac{`1}{m},v_{\mathfrak{s}}\equiv1\right)$
be a game with two players and a single location $\omega$. The dominant
action of each player is to search $\omega$; hence, both players
searching $\omega$ is the unique equilibrium, and this equilibrium
is not socially optimal (since in any cost-inclusive socially optimal
profile, only one player searches in $\omega$).
\end{example}
The players in the above example ``step on each other's toes'' because
the information structure does not allow them to each search in a
separate location. We now show that the game does admit a socially
optimal equilibrium if we assume, in addition, that the information
structure is spacious enough; more precisely, if we assume that the
information structure allows the existence of at least one redundancy-free
strategy profile, namely, a profile under which players utilize their
full search capacity and no two players search in the same location.
\begin{defn}
\label{def:redundancy-free}A strategy profile $s$ is redundancy-free
if (1) every player always uses her entire capacity (i.e., $\left|s_{i}\left(\pi_{i}\right)\right|=K_{i}$
for every cell $\pi_{i}$ of every player $i$), and (2) there is
no search duplication (i.e., $m_{s}\left(\omega\right)\leq1$ for
every $\omega\in\Omega$).
\end{defn}
\begin{prop}
\label{prop:cost-inclusive-redundancy-free}Suppose that payoffs are
ordinally consistent and solitary-search dominant, social value outweighs
cost, and there exists a redundancy-free strategy profile. Then the
game admits a (cost-inclusive) socially optimal pure equilibrium.
\end{prop}
\begin{proof}[Sketch of proof; formal proof is in Appendix \ref{subsec:Proof-of-Prop-redundanvy-free-optimal}]
\global\long\def\sTwo{s^{{\scriptscriptstyle 2}}}%
 Let $\sOne$ be a redundancy-free profile. Let $\sTwo$ be a (cost-inclusive)
socially optimal profile, and let $W$ denote the locations that are
searched (by anyone) under $\sTwo$. Start with $\sOne$; suppose
that there is a location $w\in W$ that is currently not being searched,
then switch the player who searched $w$ under $\sTwo$ from one of
her current locations to $w$. We repeat this step until we get a
profile $s$ under which the whole $W$ is searched, while $s$ is
still redundancy-free. 

First, we demonstrate that $s$ is socially optimal, by comparing
it to $\sTwo$. If there are locations outside $W$ that are searched
under $s$, then, by our assumption that social value outweighs cost,
the expected social payoff yielded by $s$ is at least as high as
that yielded by $\sTwo$ (i.e., the additional social value of these
locations exceeds the additional cost). Next, starting from $s$,
consider the improvement path of Proposition \ref{pro:sequence-unilateral-improvements}
that ends in an equilibrium. Since $s$ is redundancy-free, solitary-search
dominance implies that, along this path, no player switches either
to searching an occupied location or to being idle. A switch to an
unoccupied location that improves a player's payoff does not decrease
the social payoff, by ordinal consistency. Therefore, since $s$ is
socially optimal, so is the final equilibrium.
\end{proof}
\begin{rem}
Adding information to a player, i.e., refining her information partition,
always weakly increases the socially optimal payoff. Note, however,
that if we begin with an information structure that admits a redundancy-free
profile and then we refine some partitions, the structure becomes
more crowded, and we could end up with a structure that does not admit
a redundancy-free profile and, therefore, the new socially optimal
payoff may not be achieved by an equilibrium.
\end{rem}
\textcolor{brown}{If there is no redundancy-free profile, then the
existence of a socially optimal equilibrium is still guaranteed if,
in addition to the other assumptions of Proposition \ref{prop:cost-inclusive-redundancy-free},
the expected reward in any location is smaller, when the reward is
shared with others, than any marginal cost. The proof, then, is similar
to that of Theorem \ref{thm:balanced}. }

\section{Conclusion\label{sec:Discussion}}

Our paper studies search games in which agents explore different routes
to making a discovery that would benefit both society and the discoverer
(although the private gain may differ from the social gain). Our main
departure from the related literature is that we introduce asymmetric
information to this setup. That is, we allow each agent to have private
information about the plausibility of different routes, while almost
all of the existing literature assumes that all agents have the same
information. We believe that this is a natural development, as asymmetric
information is a key component in many real-life decentralized research
situations. In addition, we allow substantial heterogeneity between
the different routes (i.e., different expected values of finding the
prize in different locations). We also allow heterogeneity in the
rewards and costs of different players.

Our model is simplified in one aspect, as we assume that the search
is a one-shot game, while the dynamic aspects of the search interaction
are a key component in many of the existing models (see, e.g., \citealp{chatterjee2004rivals,akcigit2015role,bryan2017direction}).
While a one-shot game might model reasonably well situations with
severe time constraints, such as the motivating example of developing
a COVID-19 vaccine as soon as possible, we think that incorporating
asymmetric information in a dynamic search game is an important avenue
for future research. 

Our first main result (Theorem \ref{thm:balanced}) states that a
search game admits a (pure) equilibrium that yields the first-best
social payoff if for any two locations within a player's cell: (1)
the player and society have the same ordinal ranking over these two
locations (ordinal consistency), and (2) the player always prefers
searching in one of these locations alone to searching in the other
location with other players, or to not searching at all (solitary-search
dominance). \citet{taylor1995digging,fullerton1999auctionin,che2003optimal,koh2017incentive}
present setups of innovation contests in which it is socially optimal
to restrict the number of participating players, because adding a
player decreases the incentive of others to exert costly effort. By
contrast, Theorem \ref{thm:balanced} implies that adding players
to a search game with ordinally consistent and solitary-search dominant
payoffs always improves the maximal social payoff that can be yielded
by an equilibrium. This is so because the first-best social payoff
is (weakly) increasing when players are added. Thus, when the payoffs
are ordinally consistent and solitary-search dominant, it is socially
optimal to allow all players to participate. It is an open question
whether this property holds in our setup when we relax the assumptions
of ordinal consistency and solitary-search dominance.

Our second main result provides useful, and often binding, lower bounds
for the first-best payoff, whether it is achievable by an equilibrium
or not. These bounds are established by showing that the outcomes
that can be induced by a lottery over strategy profiles coincide with
the compatible outcomes and with the outcomes that can be generated
by a fractional allocation.

\appendix

\section*{Appendix: \label{sec:Proofs}Formal Proofs}

\section{\label{subsec:Proof-of-weakly-acyclic}Proof of Prop. \ref{pro:sequence-unilateral-improvements}
(Search Games are Weakly Acyclic)}

Given a strategy profile $s$, we define, for any cell $\pi_{i}\in\Pi_{i}$
of player $i$, the payoff of $\pi_{i}$ as the (interim) expected
payoff of player $i$ given that her signal is $\pi_{i}$, i.e., $u_{i}\left(s|\pi_{i}\right)=\sum_{\omega\in\pi_{i}}\,\mu\left(\omega|\pi_{i}\right)u_{i}\left(s|\omega\right)$.
Note that player $i$ is best-responding iff every cell of hers is
best-responding.

Player $i$ has $K_{i}$ units of capacity, which we index by numbers
between $1$ and $K_{i}$. A cell-unit of player $i$ is a pair $\left(\pi_{i},j\right)$
where $\pi_{i}\in\Pi_{i}$ is a cell, and $1\leq j\leq K_{i}$ is
a unit index. We can assume w.l.o.g. that a strategy specifies not
only in which locations within $\pi_{i}$ to search, but also which
specific cell-unit is assigned to each of these locations. Thus, for
every cell-unit $\alpha$ of $i$, a strategy of $i$ chooses a location
or chooses that $\alpha$ be inactive. We can think of a player as
being composed of many ``smaller'' decision makers, one for each
cell of hers, and of each cell as being composed of even smaller decision
makers, one for every cell-unit of that cell (with the restriction
that two cell-units of the same cell cannot search in the same location).
We define the payoff of cell-unit $\alpha=\left(\pi_{i},j\right)$
of player $i$ as the payoff of its cell $u_{i}\left(s|\pi_{i}\right)$.
Thus, every cell-unit of $\pi_{i}$ gets the same payoff. Note that
the expected reward of a cell-unit located at $\omega$ is $\mu\left(\omega|\pi_{i}\right)\cdot v_{i}^{m_{s}\left(\omega\right)}\left(\omega\right)$,
and $u_{i}\left(s|\pi_{i}\right)$ equals the sum of the expected
rewards of the active cell-units of $\pi_{i}$ minus the cost $c_{i}$
of the number of active cell-units of $\pi_{i}$.

Observe that if an inactive cell-unit $\left(\pi_{i},j\right)$ switches
to searching in location $\omega$, it makes the activation of another
cell-unit $\left(\pi_{i},j'\right)$ at $\omega'$ (weakly) less attractive
than it previously was, because of increasing marginal costs (namely,
the convexity of costs). Similarly, deactivating a cell-unit makes
a second deactivation weakly less attractive.

Given a strategy profile $s$, if there exists a deviation of a single
cell-unit that improves its own payoff then, by definition, it is
also an improvement for its cell. Conversely, let us show that the
existence of an improvement for a cell implies the existence of an
improvement for some cell-unit. Suppose first that the cell improvement
consists merely of changing the location of a few cell-units, without
changing the number of active units. Then the cost remains unchanged,
but the overall expected reward has increased. Hence, there must be
at least one cell-unit $\alpha$ whose expected reward has increased
by switching from its location $\omega$ to another location $\omega'$
that was not chosen by player $i$ under $s$. Therefore, switching
the location of $\alpha$ from $\omega$ to $\omega'$ is an improvement
for $\alpha$. Next, suppose that activating multiple cell-units is
an improvement. Then there must also exist an improvement consisting
of activating only one of these cell-units, due to the above observation
about convex costs. Similarly, if the cell can improve by deactivating
multiple cell-units then one of them can improve by deactivating itself. 

\global\long\def\sigmaTwo{\sigma^{{\scriptscriptstyle 2}}}%
Overall, we got that a player is best-responding iff all her cells
are best-responding iff all her cell-units are best-responding.
\begin{lem}
\label{lem:cell-unit}Suppose that $B$ is a set of cell-units (of
various players), $\alpha\notin B$ is another cell-unit, and $\sOne$
is a strategy profile under which every member of $B$ is best-responding.
Then there exists a finite sequence of cell-unit improvements $\sOne,\ldots,\sTee$
such that every member of $B\cup\{\alpha\}$ is best-responding under
$\sTee$.
\end{lem}
\begin{proof}
For convenience of description, let us imagine that all the inactive
cell-unit of all players stay in some place that we denote by $\lambda$.
The set $\Omega\cup\left\{ \lambda\right\} $ of the locations plus
$\lambda$ will be called the set of sites. With this terminology,
we can say that a strategy of player $i$ chooses a site for every
cell-unit of $i$ (and $\lambda$ is the only site where more than
one cell-unit of the same player can be placed). 

In what follows, whenever we mention cell-units, it only refers to
members of $B\cup\{\alpha\}$. Note that the site of all other cell-units
will remain fixed along the sequence.

Suppose that $\alpha$ is not best-responding in $\sOne$; otherwise
we are done. Let $\alpha$ switch from its current site $\theta^{1}$
to another site $\theta^{2}$ that is a best-response. The new strategy
profile is $\sTwo$. Now $\alpha$ is best-responding, and we claim
that any other cell-unit $\beta$ of the same cell is still best-responding.
We note first that $\beta$ is not placed in $\theta^{1}$ (since
$\beta$ was best-responding under $\sOne$), and w.l.o.g. it is also
not in $\theta^{2}$ (otherwise, it is currently best-responding,
since $\alpha$ is). Next we note that $\beta$ cannot improve by
switching to $\theta^{1}$; otherwise, simply switching $\beta$ to
$\theta^{2}$ would have been an improvement earlier, in $\sOne$.

Suppose first that $\theta^{1}$ and $\theta^{2}$ are both locations.
Then, since $\theta^{2}$ is now occupied, and the preference relation
between sites other than $\theta^{1}$ and $\theta^{2}$ has not changed
(as the cost has not changed), $\beta$ is indeed still best-responding.
Next suppose that $\theta^{2}=\lambda$. Then, by the convex costs
observation, the attractiveness of $\lambda$ has decreased by the
switch from $\theta^{1}$ to $\theta^{2}=\lambda$, hence $\beta$
still cannot improve by switching to $\lambda$. And although the
cost has changed, the relation between locations other than $\theta^{1}$
has not changed; hence, $\beta$ is best-responding. Finally, suppose
that $\theta^{1}=\lambda$. Then the relation between locations other
than $\theta^{2}$ has not changed; hence, $\beta$ is best-responding.

\paragraph{Phase I }

In $\sTwo$, we add a dummy player at the site $\theta^{1}$, denoting
the resulting strategy profile by $\sigmaTwo$ (for a profile $\sigma^{t}$
that includes the dummy player, $s^{t}$ will denote the same profile
without the dummy). Then Phase I begins: at every stage of Phase I,
one cell-unit who is currently not best-responding switches to a best-response
site. This continues as long as there are non-best-responding cell-units,
unless someone switches to $\theta^{1}$, in which case Phase I immediately
terminates.

As we will see, Phase I begins by some cell-unit switching from $\theta^{2}$
to another site $\theta^{3}$, then another cell-unit switching from
$\theta^{3}$ to another site, etc. More specifically, we claim that
under any strategy profile $\sigma=\sigma^{t}$ encountered during
Phase I,

\noindent (a) there exists exactly one site $\theta$ that is chosen
by one more cell-unit than under $\sOne$, i.e., $m_{\sigma}(\theta)=m_{\sOne}(\theta)+1$,
while for every other site $\theta'$, $m_{\sigma}(\theta')=m_{\sOne}(\theta')$
(we call $\theta$ ``the plus site''); and

\noindent (b) for any cell-unit $\beta$ whose current site is some
$\theta$ with $m_{\sigma}(\theta)>0$, and who can also choose another
site $\theta'$, if there were $m_{\sOne}(\theta)$ cell-units at
$\theta$ (including $\beta$ itself) and $m_{\sOne}(\theta')$ cell-units
at $\theta'$, then $\beta$ would weakly prefer $\theta$ to $\theta'$.

\noindent Property (b) roughly says that if $\beta$ is not best-responding,
it is only because $\beta$ is in the plus site.

When Phase I starts, in $\sigmaTwo$, property (a) holds and $\alpha$
has just switched to the plus site $\theta^{2}$. Since all cell-units
of the cell of $\alpha$ were best-responding in $s^{{\scriptscriptstyle 2}}$
(i.e., without the dummy) they obey property (b) in $\sigmaTwo$ (in
particular, $\alpha$ is currently best-responding when $\alpha$'s
current site is the plus site, let alone when it is not the plus site).
As for cell-units of other cells, they obeyed (b) in $s^{{\scriptscriptstyle 1}}$
and, therefore, they still do, as the switch of $\alpha$ or the addition
of the dummy do not affect that.

The claim is proved by induction from one stage to the next: suppose
that cell-unit $\beta$ improves on stage $t$ by switching from $\theta$
to $\theta'$. Since $\beta$ could improve, (b) implies that $\theta$
must have been the plus site in stage $t$. Therefore, the plus will
move with $\beta$ from $\theta$ to $\theta'$, and (a) will still
hold in stage $t+1$. Note that, importantly, if the plus site is
$\lambda$ in some stage then every cell-unit is best-responding,
since the number of partners does not affect the reward in $\lambda$,
which simply equals $0$; hence, Phase I will end on that stage.

As for property (b), $\beta$ best-responds in $\theta'$, and other
cell-units of $\beta$'s cell obeyed (b) on stage $t$, implying that
they were best-responding on that stage. It follows, by the same argument
we used above for $t=1$ (i.e., the transition from $s^{{\scriptscriptstyle 1}}$
to $s^{{\scriptscriptstyle 2}}$), that they also best-respond on
stage $t+1$. Therefore, they obey (b); and cell-units of other cells
still obey (b), as they were not affected by $\beta$'s switch.

To see that Phase I cannot go on forever, recall first that it ends
if the plus site is $\lambda$. Otherwise, on each stage of Phase
I some cell-unit $\beta$ switches from location $\omega$ to location
$\omega'$, and the costs always remain fixed. Since this switch is
an improvement, it strictly increases the expected reward of $\beta$.
Now $\omega'$ becomes the plus site, and afterwards the expected
reward of $\beta$ when placed in $\omega'$ can never be lower than
it is now, while it can be higher if the plus is somewhere else (or
if $\beta$ improves again).\footnote{One can verify that, in fact, $\beta$ will not switch again during
Phase I. We employ a different argument here, in order to strengthen
the analogy with Phase II.} Thus, the expected reward of $\beta$ will never go down to the level
it was at before the switch. Hence, Phase I cannot turn into a cycle,
and since there are only finitely many strategy profiles, Phase I
must eventually end.

Recall that Phase I terminates once someone switches to $\theta^{1}$.
Therefore, the plus site cannot be $\theta^{1}$ during this phase
except maybe at the end, and hence nobody switches from $\theta^{1}$
during Phase I. Therefore, all the switches are improvements not only
in the game with the dummy added at $\theta^{1}$, but also in the
original game.

Denote the strategy profile at the end of Phase I by $\sigma^{*}$.
Now we remove the dummy from $\theta^{1}.$ Suppose first that Phase
I ended because somebody has just switched to $\theta^{1}$. Then,
the removal of the dummy means that now there is no plus site at all,
and (b) implies that every cell-unit best-responds under $s^{*}$
(recall that $s^{*}$ is $\sigma^{*}$ without the dummy player),
and we are done. 

\paragraph{Phase II}

Otherwise, Phase I ended at $\sigma^{*}$ because everyone was best-responding
(when the dummy was still at $\theta^{1}$). Starting from $s^{*}$,
we define Phase II analogously to Phase I (while Phase I more or less
described a process of restabilizing the system after one cell-unit
is added, Phase II describes restabilizing it after one cell-unit
is removed), as follows. At every stage, as long as there are cell-units
who are not best-responding, choose a cell, and choose a switch of
a single cell-unit that would yield the highest increase in that cell's
payoff.

As we will see, Phase II begins by some cell-unit switching to $\theta^{1}$
from some site $\theta'$, then another cell-unit switching to $\theta'$
from another site, etc. More specifically, we claim that under any
strategy profile $s=s^{t}$ encountered during Phase II,

\noindent (a') there exists exactly one site $\theta$ with $m_{s}(\theta)=m_{\sigma^{*}}(\theta)-1$,
while for every other site $\theta'$, $m_{s}(\theta')=\hat{m}_{\sigma^{*}}(\theta')$
(we call $\theta$ ``the minus site''); and

\noindent (b') for any cell-unit $\beta$ whose current site is some
$\theta$ and who can also choose site $\theta'$, if there were $m_{\sigma^{*}}(\theta)$
cell-units at $\theta$ (including $\beta$) and $m_{\sigma^{*}}(\theta')$
cell-units at $\theta'$, then $\beta$ would weakly prefer $\theta$
to $\theta'$.

The analysis is almost analogous to that of Phase I. When Phase II
starts, in the profile $s^{*}$, (a') holds and $\theta^{1}$ is the
minus site. (b') also holds because everyone was best-responding under
$\sigma^{*}$. The claim is proved by induction from one stage to
the next: suppose that cell-unit $\beta$ improves on stage $t$ by
switching from $\theta$ to $\theta'$. Improvement implies, by (b'),
that $\theta'$ must have been the minus site in stage $t$. Therefore,
the minus will move from $\theta'$ to $\theta$, and (a') will still
hold in stage $t+1$. Note that if the minus site is $\lambda$ in
some stage, then every cell-unit is best-responding and Phase II will
end on that stage.

Let $\pi_{i}$ be the cell of $\beta$. Any cell-unit of another cell
still obeys (b'), as it was not affected by $\beta$\textquoteright s
switch. Since the switch of $\beta$ from $\theta$ to $\theta'$
was, by definition of Phase II, a best cell-unit switch for $\pi_{i}$,
$\beta$ cannot improve again. Therefore, $\beta$ obeys (b'), since
$\beta$ is not in the minus site. Let $\gamma$ be another cell-unit
of $\pi_{i}$. Note first that $\gamma$ cannot improve by switching
to the current minus site $\theta$; otherwise, switching $\gamma$
to $\theta'$ earlier, on stage $t$, would have been a better switch
than the one chosen. 

We have seen that $\theta'\neq\lambda$. If also $\theta\neq\lambda$,
then the preference relation between sites other than $\theta$ and
$\theta'$ has not changed, and, therefore, $\gamma$ still obeys
(b'). Otherwise, $\theta=\lambda$. Then if $\gamma$ is placed in
$\lambda$ it still obeys (b'), because the attractiveness of $\lambda$
has increased, by the convex costs observation; and if $\gamma$ is
placed in some location then $\gamma$ still obeys (b'), because the
relation between locations other than $\theta'$ has not changed.

When Phase II ends, every cell-unit will be best-responding. To see
that Phase II must eventually end, we employ the same argument as
in Phase I, noting that right after some cell-unit $\beta$ switches
to location $\omega'$, $\omega'$ is not the minus site, and, therefore,
the expected reward of $\beta$ when placed in $\omega'$ can never
be lower than it is now.
\end{proof}
To prove weak acyclicity, start from any strategy profile. By applying
Lemma \ref{lem:cell-unit} inductively we obtain a sequence of cell-unit
improvements that lead to a profile under which one cell-unit is best-responding,
then two, and so on. Eventually we get a profile under which every
cell-unit is best-responding, hence an equilibrium.

\section{\label{subsec:Proof-of-Claim-prob-to-smooth}Proof of Prop. \ref{prop:equivalent-outcome-classes}
(3 Equivalent Classes of Outcomes) }

\global\long\def\Reals{\mathbb{{R}}}%

We already explained why (i)$\Rightarrow$(ii), right before Definition
\ref{def:compatibility}. To see that (ii)$\Rightarrow$(iii), suppose
$f$ is compatible. Denote the set of all cells of the search game
by $\hat{\Pi}=\{(i;\pi_{i}):i\in N,\pi_{i}\in\Pi_{i}\}$. We construct
a flow network, namely, a directed graph $D=(V,E)$ with vertices
$V$ and edges $E\subset V\times V$, and a flow capacity $\kappa(v_{{\scriptscriptstyle 1}},v_{{\scriptscriptstyle 2}})\geq0$
\,for every edge $(v_{{\scriptscriptstyle 1}},v_{{\scriptscriptstyle 2}})$
(illustrated beside the sketch of this proof, in Figure \ref{fig:min-cut-1}).
There are two special vertices, a source $s$ and a sink $t$. The
other vertices in our network are the locations $\Omega$ and the
cells $\hat{\Pi}$ of the game. There is an edge from $s$ to every
$(i;\pi_{i})\in\hat{\Pi}$, where $\kappa(s,(i;\pi_{i}))=K_{i}$,
and an edge from every location $\omega\in\Omega$ to $t$, where
$\kappa(\omega,t)=f(\omega)$. Also, there is an edge from a cell
$(i;\pi_{i})\in\hat{\Pi}$ to a location $\omega$ iff $\pi_{i}$
contains $\omega$, and the flow capacity $\kappa$ of such edges
is infinite (for a textbook presentation of flow networks; see, e.g.,
\citealp[Ch. 26]{cormen2009introduction}).

A cut of $D$ is a subset of edges $C\subset E$, such that if all
the edges of $C$ are removed then there exists no path between $s$
and $t$. Suppose that $C$ is a minimal cut, i.e., a cut whose sum
of capacities is minimal. Then $C$ certainly does not include any
edge between a cell and a location, as those edges have an infinite
flow capacity. Let $Q=\left\{ \omega\in\Omega:\left(\omega,t\right)\in C\right\} $
denote the locations that the cut separates from $t$. Denote $W=\Omega\setminus Q$.
Then $C$ must include all the edges $\{(s,(i;\pi_{i})):i\in N,\pi_{i}\cap W\neq\emptyset\}$;
otherwise there would still exist a path from $s$ to $t$. Hence,
the total capacity of $C$ equals 
\[
\sum_{\omega\in Q}\kappa\left(\omega,t\right)+\sum_{i\in N}\sum_{\pi_{i}\cap W\neq\emptyset}\kappa\left(s,\left(i;\pi_{i}\right)\right)=\sum_{\omega\in Q}f(\omega)+\sum_{i\in N}\sum_{\pi_{i}\cap W\neq\emptyset}K_{i}=
\]
\[
\sum_{\omega\in Q}f(\omega)+\sum_{i\in N}K_{i}\cdot\left|\left\{ \pi_{i}\in\Pi_{i}:\pi_{i}\cap W\neq\emptyset\right\} \right|\geq\sum_{\omega\in Q}f(\omega)+\sum_{\omega\in W}f(\omega)=\sum_{\omega\in\Omega}f(\omega).
\]
Therefore, the cut that consists of all edges of type $(\omega,t)$,
whose total capacity equals $\sum_{\omega\in\Omega}f(\omega)$, is
minimal.

A flow in $D$ is a function $\varphi:E\to\Reals^{+}$ such that:
(i) the flow never exceeds the capacity, i.e., $\varphi(e)\leq\kappa(e)$,
and (ii) the overall flow outgoing from $s$, namely, the sum of flows
on edges outgoing from $s$, equals the overall flow incoming to $t$,
namely, the sum of flows on edges incoming to $t$ (call this quantity
the value of the flow), and for any other vertex the incoming flow
equals the outgoing flow. The max-flow min-cut theorem (\citealp[p. 723, Theorem 26.6]{cormen2009introduction})
states that the value of the maximal flow equals the total capacity
of the minimal cut; therefore, $D$ admits a flow $\varphi$ of value
$\sum_{\omega\in\Omega}f(\omega)$, and so it must be the case that
$\varphi(\omega,t)=f(\omega)$ for every $\omega\in\Omega$. 

Now define a fractional allocation $\alpha$ by letting ${\color{brown}\alpha}_{i}(\pi_{i},\omega)=\varphi((i;\pi_{i}),\omega)$
for every $i\in N$,$\pi_{i}\in\Pi_{i}$, and $\omega\in\pi_{i}$.
To see that this is a fractional allocation we verify that for any
$\pi_{i}$,\, $\sum_{\omega\in\pi_{i}}{\color{brown}\alpha}_{i}\left(\pi_{i},\omega\right)=\sum_{\omega\in\pi_{i}}\varphi((i;\pi_{i}),\omega)=\varphi(s,(i;\pi_{i}))\leq\kappa(s,(i;\pi_{i}))=K_{i}$
(where the second equality is due to the equality of the outgoing
and the incoming flow). To see that $\alpha$ generates $f$, we verify
that for any $\omega$, it is the case that $\sum_{i\in N}\,{\color{brown}\alpha}_{i}\left(\pi_{i}(\omega),\omega\right)=\sum_{i\in N}\,\varphi\left((i;\pi_{i}(\omega)),\omega\right)=\varphi(\omega,t)=f(\omega)$.

Now we show that (iii)$\Rightarrow$(i). A nonnegative matrix $\boldsymbol{A}$
is doubly stochastic (resp., doubly substochastic) if the sum of the
elements in each row and in each column is equal to (resp., at most)
one, i.e., if $\sum_{j}A_{ij}=1$ (resp., $\sum_{j}A_{ij}\leq1$)
for each row $i$ and $\sum_{i}A_{ij}=1$ (resp., $\sum_{i}A_{ij}\leq1$)
for each column $j$. Note that any doubly stochastic matrix must
be a square matrix (but this is not the case for a doubly substochastic
matrix). A doubly stochastic (resp., doubly substochastic) matrix
is a permutation (resp., subpermutation) matrix if it includes only
zeros and ones, i.e., if $A_{ij}\in\left\{ 0,1\right\} $ for any
$i,j$. Note that a permutation (resp., subpermutation) matrix includes
exactly (resp., at most) one non-zero value in each row and in each
column, and this value is equal to one. The Birkhoff--von Neumann
theorem states that any doubly stochastic matrix can be written as
a convex combination of permutation matrices. Formally:
\begin{thm}[Birkhoff--von Neumann Theorem]
\label{thm:-Birkhoff} Let $\boldsymbol{A}$ be a doubly stochastic
matrix. Then there exists a finite set of permutation matrices $\boldsymbol{P}^{1}$,
...,$\boldsymbol{P}^{K}$ such that $\boldsymbol{A}=\sum_{k}w_{k}\cdot\boldsymbol{P}^{k}$,
where $w_{k}\geq0$ for each $k$ and $\sum_{k}w_{k}=1$.
\end{thm}
We present a simple extension of Thm. \ref{thm:-Birkhoff} \,that
states that any doubly substochastic matrix can be written as a convex
combination of subpermutation matrices.\footnote{One can show that Lemma \ref{lem:extended-Birkhoff} is implied by
the extension of the Birkhoff--von Neumann Theorem presented in \citet{budish2013designing}.
For completeness, we provide a self-contained proof of the lemma.} 
\begin{lem}
\label{lem:extended-Birkhoff}Let $\boldsymbol{A}$ be a doubly substochastic
matrix. Then there exists a finite set of subpermutation matrices
$\boldsymbol{Q}^{1}$, ...,$\boldsymbol{Q}^{K}$ s.t. $\boldsymbol{A}=\sum_{k}w_{k}\cdot\boldsymbol{Q}^{k}$,
where $w_{k}\geq0$ for each $k$ and $\sum_{k}w_{k}=1$.
\end{lem}
\begin{proof}
Let $I$ (resp., $J$) be the number of rows (resp., columns) in the
matrix $\boldsymbol{A}$. We construct a square doubly stochastic
matrix $\boldsymbol{B}$ with $I+J$ rows and columns by merging 4
submatrices (as illustrated in Figure \ref{fig:Illustration-Birkhoff}):
(1) the matrix $\boldsymbol{A}$ (with $I$ rows and $J$ columns)
in the top-left part of $\boldsymbol{B}$, (2) a $J\times J$ diagonal
matrix in the bottom-left part of $\boldsymbol{B}$, where each diagonal
cell completes the values in each column of $\boldsymbol{A}$ to one,
(3) an $I\times I$ diagonal matrix in the top-right part of $\boldsymbol{B}$,
where each diagonal cell completes the values in each row of $\boldsymbol{A}$
to one, and (4) the $J\times I$ matrix $\boldsymbol{A}^{T}$ (the
transpose of $\boldsymbol{A}$) in the bottom-right part of $\boldsymbol{B}$.
It is immediate that $\boldsymbol{B}$ is a doubly stochastic matrix.
By Theorem \ref{thm:-Birkhoff} there exists a finite set of permutation
matrices $\boldsymbol{P}^{1}$, ...,$\boldsymbol{P}^{K}$ (with $I+J$
rows and columns) such that $\boldsymbol{B}=\sum_{k}w_{k}\cdot\boldsymbol{P}^{k}$,
where $w_{k}\geq0$ for each $k$ and $\sum_{k}w_{k}=1$. Let $\boldsymbol{Q}^{k}$
be a submatrix of $\boldsymbol{P}^{k}$ with the first $I$ rows and
$J$ columns. Then it is immediate that each $\boldsymbol{Q}^{k}$
is a subpermutation matrix and that $\boldsymbol{A}=\sum_{k}w_{k}\cdot\boldsymbol{Q}^{k}$.
\begin{figure}[h]
\caption{\label{fig:Illustration-Birkhoff}Illustration of How to Construct
the Square Matrix $\boldsymbol{B}$}

\medskip{}

\centering{}\includegraphics[scale=0.35]{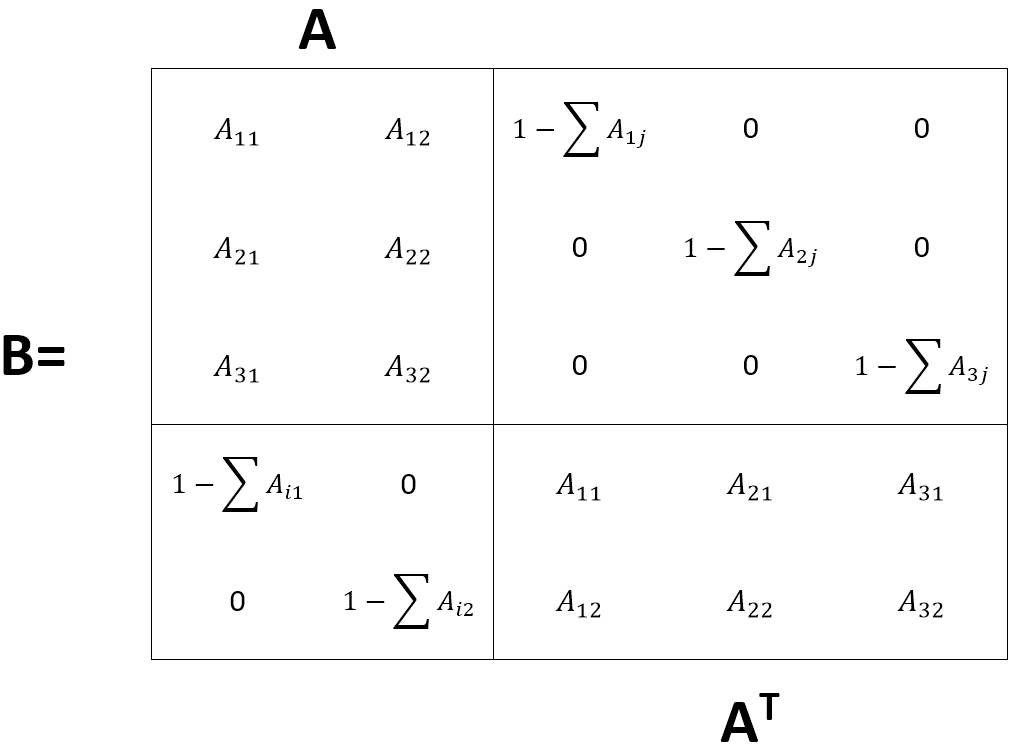}
\end{figure}
\end{proof}
Next we rely on Lemma \ref{lem:extended-Birkhoff} to prove that (iii)$\Rightarrow$(i).
Suppose $f$ is generated by the fractional allocation $\alpha$.
Similarly to the proof of Proposition \ref{pro:sequence-unilateral-improvements},
we define a cell-unit as a tuple $\left(i,j,\pi_{i}\right)$, where
$i\in N$ is a player, $j\in\left\{ 1,...,K_{i}\right\} $ is an index
corresponding to one unit of capacity of player $i$, and $\pi_{i}\in\Pi_{i}$
is a cell of player $i$. Let $\hat{\Pi}$ denote the set of all cell-units
with a typical element $\hat{\pi}$, let $\hat{\Pi}_{i}$ denote the
subset of cell-units that correspond to player $i$, and let $\hat{\Pi}_{i,j}$
denote the subset of cell-units that correspond to capacity unit $j\in\left\{ 1,...,K_{i}\right\} $
of player $i$. We write $\omega\in\hat{\pi}=\left(i,j,\pi\right)$
if $\omega\in\pi$.

A fractional division (of the cell-units) $\tau$ allocates, for each
cell-unit $\hat{\pi}$, a capacity of one between its locations, i.e.,
it specifies a nonnegative number $\tau\left(\hat{\pi},\omega\right)$
for every location $\omega\in\hat{\pi}$, such that $\sum_{\omega\in\hat{\pi}}\,\tau\left(\hat{\pi},\omega\right)\leq1$.
The fractional allocation $\alpha$ can be represented as an equivalent
fractional division (of the cell-units) $\tau$ that satisfies $\sum_{j=1}^{K_{i}}{\color{brown}\tau}\left(\left(i,j,\pi_{i}\right),\omega\right)={\color{brown}\alpha}_{i}\left(\pi_{i},\omega\right)$
for each $\pi_{i}\in\Pi_{i}$ and $\omega\in\Omega$. The equivalent
fractional division $\tau$ can be represented as a $\left|\hat{\Pi}\right|\times\left|\Omega\right|$
nonnegative matrix $\boldsymbol{C}$ as follows:

\[
C_{\left(i,j,\pi_{i}\right),\omega}=\begin{cases}
{\color{brown}\tau}\left(\left(i,j,\pi_{i}\right),\omega\right) & \omega\in\pi_{i}\in\Pi_{i}\\
0 & \textrm{otherwise}.
\end{cases}
\]
Observe that the sum of each row in $\boldsymbol{C}$ is at most one,
i.e., $\sum_{\omega\in\Omega}C_{\hat{\pi},\omega}\leq1$, but the
sum of a column might be greater than one. Let $\boldsymbol{A}$ be
the matrix derived from $\boldsymbol{C}$ by decreasing the values
of the lower cells within columns whose sum is greater than one, such
that the sum of each column is at most one. Formally (where we write
$\hat{\pi}'<\hat{\pi}$ if the row of $\hat{\pi}'$ is higher than
the row of $\hat{\pi}$ in the matrix $\boldsymbol{C}$): 
\[
A_{\hat{\pi},\omega}=\begin{cases}
C_{\hat{\pi},\omega} & \sum_{\hat{\pi}'\leq\hat{\pi}}C_{\hat{\pi}',\omega}\leq1\\
1-\sum_{\hat{\pi}'<\hat{\pi}}C_{\hat{\pi}',\omega} & \sum_{\hat{\pi}'<\hat{\pi}}C_{\hat{\pi}',\omega}\leq1<\sum_{\hat{\pi}'\leq\hat{\pi}}C_{\hat{\pi}',\omega}\\
0 & \sum_{\hat{\pi}'<\hat{\pi}}C_{\hat{\pi}',\omega}>1.
\end{cases}
\]
Observe that $\boldsymbol{A}$ is a doubly substochastic matrix (i.e.,
the sum of each row and of each column is at most one), and that the
fractional division corresponding to $\boldsymbol{A}$ generates the
same mixed outcome as $\alpha$. By Lemma \ref{lem:extended-Birkhoff},
there exists a finite set of subpermutation matrices $\boldsymbol{Q}^{1},\ldots,\boldsymbol{Q}^{K}$
such that $\boldsymbol{A}=\sum_{k}w_{k}\cdot\boldsymbol{Q}^{k}$,
where $w_{k}\geq0$ for each $k$ and $\sum_{k}w_{k}=1$. Further
observe that each subpermutation matrix $\boldsymbol{Q}^{k}$ corresponds
to the cell-unit representation of a pure strategy profile $s^{k}$,
which implies that ${\color{brown}\alpha}$ generates the same mixed
outcome as the lottery over strategy profiles $\sigma=\sum_{k}w_{k}\cdot s^{k}$.

\section{\label{subsec:Proof-of-Prop-redundanvy-free-optimal}Proof of Proposition
\ref{prop:cost-inclusive-redundancy-free} (Cost-inclusive Variant)}

The game admits a redundancy-free strategy profile. Let $\sOne$ be
some redundancy-free profile. Let $\sTwo$ be a cost-inclusive socially
optimal profile, and let $W\subseteq\Omega$ denote the locations
that are searched (by anyone) under $\sTwo$. First, we claim that
there exists a redundancy-free profile $s$ such that every location
in $W$ is searched under $s$.

\global\long\def\mOne{m^{{\scriptscriptstyle 1}}}%
\global\long\def\mTwo{m^{{\scriptscriptstyle 2}}}%

To prove this claim, consider a bipartite graph whose left side is
the set $A$ that consists of $K_{i}$ copies\footnote{The members of $A$ correspond to the cell-units defined in the proof
of Proposition \ref{pro:sequence-unilateral-improvements}.} of each cell of each player $i$, and whose right side is the set
of locations $\Omega$. Two nodes $a\in A$ and $\omega\in\Omega$
(i.e, a copy of a cell and a location) are connected by an edge iff
the cell contains that location. For every cell of player i, a strategy
$s_{i}$ corresponds to a list of pairs of nodes $\left(a,\omega\right)$,
where $a$ is a copy of that cell and $\omega$ is chosen by $s_{i}$.
The strategy profile $\sOne$ avoids search duplication and, therefore,
it corresponds to a matching $\mOne$ in the graph, namely, a list
of pairs of nodes $\left(a,\omega\right)$ such that each pair is
connected and no node appears twice. Moreover, $\sOne$ employs every
unit of capacity and, therefore, the corresponding matching $\mOne$
fully matches $A$; i.e., it matches every node of $A$. Similarly,
the strategy profile $\sTwo$ corresponds to a matching $\mTwo$ that
fully matches the subset of nodes $W\subseteq\Omega$, where for any
location that is searched by more than one player under $\sTwo$,
we arbitrarily pick one of these players and discard the others.\footnote{Alternatively, we can assume w.l.o.g. that $\sTwo$ involves no search
duplications, as $\sTwo$ is socially optimal.}

Thus, in our bipartite graph there is one matching that fully matches
the left side $A$, and another matching that fully matches a subset
$W\subseteq\Omega$ of the right side. We now claim that this implies
the existence of a single matching that fully matches both $A$ and
$W$. The proof is by induction on the size of $W$; note that if
$W$ is empty then the claim is trivial (this is the induction's base). 

Suppose that $\mOne$ does not fully match $W$ (otherwise we are
done). Then, let $w\in W$ be some node that is left unmatched by
$\mOne$, and let $a\in A$ be its match according to $\mTwo$. Let
us remove the nodes $a$ and $w$ from the graph, and look at the
residual graph whose left side is $A'=A\setminus\left\{ a\right\} $
and whose right side is $\Omega\setminus\left\{ w\right\} $. The
restriction of $\mTwo$ to the residual graph fully matches the set
$W'=W\setminus\left\{ w\right\} $. Also, the restriction of $\mOne$
to the residual graph fully matches the set $A'$, as the removal
of $w$ does not affect $\mOne$. By applying the induction hypothesis
to the residual graph, there exists a matching that fully matches
both $A'$ and $W'$. This matching, together with the matching of
$a$ and $w$, gives us a matching on the whole graph, which fully
matches both $A$ and $W$. 

Getting back to out setting, our claim about a matching on the graph
translates to the existence of a strategy profile $s$ that is both
redundancy-free and searches all the locations in $W$. Let us show
that this $s$ is socially optimal. Let $k$ be the number of idle
units of capacity under $\sTwo$. Under $s$, every unit of capacity
is employed and therefore the additional cost, compared to $\sTwo$,
is the sum of the marginal costs of these $k$ units. On the other
hand, since $s$ is redundancy-free there are at least $k$ additional
locations, besides $W$, that are searched under $s$. Our assumption
that social value outweighs cost implies that the expected social
gain by the additional locations exceeds the additional cost. Therefore,
the social payoff under $s$ is at least as much as that under $\sTwo$,
implying that $s$ is socially optimal.

By Proposition \ref{pro:sequence-unilateral-improvements}, there
exists an improvement path, starting from $s,$ where in each stage
some player improves her payoff by switching a single choice in a
single cell, and ending in an equilibrium. Since $s$ is redundancy-free,
the whole path also consists of redundancy-free strategy profiles,
because neither switching to searching an occupied location nor switching
to being idle can be an improvement, by the solitary-search dominance
assumption. By ordinal consistency, whenever a player improves her
payoff by switching from one (unoccupied except by her) location to
another unoccupied location, the social payoff does not decrease.
Therefore, the equilibrium reached at the end of the path is still
socially optimal.

\bibliographystyle{chicago}
\bibliography{diffG_Refs}

\end{document}